\documentclass[10pt]{article}
\usepackage{times}               
\usepackage{amsmath,amssymb,amsthm,url}
\usepackage{graphicx}
\usepackage{epstopdf}
\def\today{\number\day\ \ifcase\month\or
 January\or February\or March\or April\or May\or June\or
 July\or August\or September\or October\or November\or December\fi
 \space \number\year}
\def\gobble#1#2{}
\def\shortdate{\number\day/\number\month/\expandafter\gobble\number\year}

\def\N{\mathbb{N}}
\def\R{\mathbb{R}}
\def\p{Pain\-lev\'e}
\def\bk{B\"ack\-lund}

\def\bts{B\"ack\-lund transformations}

\def\peq{\p\ equation}
\def\peqs{\p\ equations}
\def\PI{\mbox{\rm P$_{\rm I}$}}
\def\PII{\mbox{\rm P$_{\rm II}$}}
\def\PIII{\mbox{\rm P$_{\rm III}$}}
\def\PIV{\mbox{\rm P$_{\rm IV}$}}
\def\PV{\mbox{\rm P$_{\rm V}$}}
\def\PVI{\mbox{\rm P$_{\rm VI}$}}

\def\sPIV{\hbox{\rm S$_{\rm IV}$}}

\def\det{\mathop{\rm det}\nolimits}
\def\erf{\mathop{\rm erf}\nolimits}
\def\erfc{\mathop{\rm erfc}\nolimits}
\def\erfi{\mathop{\rm erfi}\nolimits}

\def\ds{\displaystyle}

\def\Ai{\mathop{\rm Ai}\nolimits}
\def\Bi{\mathop{\rm Bi}\nolimits}

\newcommand{\BesselJ}[1]{J_{#1}}
\newcommand{\BesselY}[1]{Y_{#1}}

\newcommand{\WhitD}[1]{D_{#1}}
\newcommand{\WhitM}[2]{M_{#1,#2}}
\newcommand{\WhitW}[2]{W_{#1,#2}}
\newcommand{\KummerM}{M}
\newcommand{\KummerU}{U}

\newcommand{\HyperpFq}[2]{{}_{#1}F_{#2}}
\newcommand{\JacobiP}[3]{P^{(#1,#2)}_{#3}}
\newcommand{\HermiteH}[1]{H_{#1}}

\newcommand{\LaguerreL}[2]{L^{(#1)}_{#2}}

\def\d{{\rm d}}\def\e{{\rm e}}\def\i{{\rm i}}
\def\a{\alpha}
\def\b{\beta}
\def\k{\kappa}
\def\w{\omega}
\def\ep{\varepsilon}
\def\th{\vartheta}
\def\vth{\vartheta}
\def\ph{\varphi}

\newcommand{\Integer}{\mathbb{Z}}

\def\Z{\Integer}

\def\W{\mathcal{W}}

\def\ode{ordinary differential equation}
\def\odes{ordinary differential equations}

\newcommand{\deriv}[3][]{\frac{\d^{#1}{#2}}{{\d{#3}}^{#1}}}
\newcommand{\pderiv}[3][]{\frac{\partial^{#1}{#2}}{{\partial{#3}}^{#1}}}

\newtheorem{theorem}{Theorem}[section]

\newtheorem{lemma}[theorem]{Lemma}
\newtheorem{corollary}[theorem]{Corollary}
\theoremstyle{definition}

\newtheorem{remark}[theorem]{Remark}
\newtheorem{remarks}[theorem]{Remarks}
\numberwithin{figure}{section}
\numberwithin{equation}{section}
\numberwithin{table}{section}

\def\Eref#1{Equation (\ref{#1})}
\def\sref#1{\S\ref{#1}}

\newcommand{\comment}[1]{}

\def\qz{\deriv{q}{z}}
\def\qzz{\deriv[2]{q}{z}}
\def\la{\lambda}
\def\intS{\int_a^b}
\def\dotsn{\!\!\!\begin{array}{c}{\cdots}\\[-10pt]{}_{n}\end{array}\!\!\!}

\begin{document}
\title{The relationship between semi-classical Laguerre polynomials\\
and the fourth \p\ equation}
\author{Peter A.\ Clarkson\\ School of Mathematics, Statistics and Actuarial Science,\\
University of Kent, Canterbury, CT2 7NF, UK\\
Email: \texttt{P.A.Clarkson@kent.ac.uk}\\[10pt]
and\\[10pt] Kerstin Jordaan\\ Department of Mathematics and Applied Mathematics,\\
University of Pretoria, Pretoria, 0002, South Africa\\
Email: \texttt{kerstin.jordaan@up.ac.za}}
\maketitle

\begin{abstract}
We discuss the relationship between the recurrence coefficients of orthogonal polynomials with respect to a
semi-classical Laguerre weight and classical solutions of the fourth \p\ equation. We show that the coefficients in these recurrence relations can be expressed in terms of Wronskians of parabolic cylinder functions which arise in the description of special function solutions of the fourth \p\ equation.
\end{abstract}

\begin{itemize}
\item[]\textbf{Keywords}: Semi-classical orthogonal polynomials; Recurrence coefficients; \p\ equations; Wronskians; Parabolic cylinder functions; Hamiltonians
\item[]\textbf{Mathematics Subject Classification (2010)}: Primary 34M55, 33E17, 33C47; Secondary 42C05, 33C15
\end{itemize}

\section{Introduction}
In this paper we are concerned with the coefficients in the three-term recurrence relations for orthogonal polynomials with respect to the semi-classical Laguerre weight
\begin{equation}\w(x;t)=x^\la\exp(-x^2+tx),\qquad x\in\R^+,\label{eq:scLag1}\end{equation}
with parameters $\la>-1$ and $t\in\R$, which has been recently studied by Boelen and van Assche \cite{refBvA} and Filipuk, van Assche and Zhang \cite{refFvAZ}. It is shown that these recurrence coefficients can be expressed in terms of 
Wronskians that arise in the description of special function solutions of the fourth \p\ equation (\PIV)
\begin{equation}\label{eq:PIV} 
\deriv[2]{q}{z}= \frac{1}{2q}\left(\deriv{q}{z}\right)^{2} + \frac{3}{2}q^3 + 4z q^2 + 2(z^2 - A)q + \frac{B}{q},
\end{equation}
where $A$ and $B$ are constants, which are expressed in terms of parabolic cylinder functions.

The relationship between semi-classical orthogonal polynomials and integrable equations dates back to the work of Shohat \cite{refShohat} and later Freud \cite{refFreud}, as well as Bonan and Nevai \cite{refBNevai}. However it was not until the work of Fokas, Its and Kitaev \cite{refFIKa,refFIKb} that these equations were identified as discrete \p\ equations. The relationship between semi-classical orthogonal polynomials and the (continuous) \p\ equations was demonstrated by Magnus \cite{refMagnus95,refMagnus99} who showed that the coefficients in the three-term recurrence relation for the Freud weight \cite{refBNevai,refFreud,refvA07}
\[\omega(x;t)=\exp\left(-\tfrac14x^4-tx^2\right),\qquad x\in\R,\] with $t\in\R$ a parameter,
can be expressed in terms of solutions of \PIV\ (\ref{eq:PIV}).

A motivation for this work is the fact that recurrence coefficients of semi-classical orthogonal polynomials can often be expressed in terms of solutions of the \peqs. For example, recurrence coefficients are expressed in terms of solutions of \PII\ for semi-classical orthogonal polynomials with respect to the Airy weight
\[\w(x;t)=\exp\left(\tfrac13x^3+tx\right),\qquad x^3<0,\]with $t\in\R$ a parameter \cite{refMagnus95}; 
in terms of solutions of \PIII\ for the perturbed Laguerre weight
\[\w(x;t)=x^\a\exp\left(-x-{t}/{x}\right),\qquad x\in\R^+,\] 
with $\a>0$ and $t\in\R^+$ parameters \cite{refCI10}; 
in terms of solutions of \PV\ for the weights
\[\begin{array}{l@{\qquad}l}
\w(x;t)=(1-x)^{\a}(1+x)^{\b}\e^{-tx},& x\in[-1,1],\\
\w(x;t)=x^\a(1-x)^\b\e^{-t/x},& x\in[0,1],\\
\w(x;t)=x^{\a}(x+t)^{\b}\e^{-x},& x\in\R^+,
\end{array}\] 
with $\a,\b>0$ and $t\in\R^+$ parameters \cite{refBCE10,refBCM,refCD10,refCM12,refFW07}; and in terms of solutions of \PVI\ for the generalized Jacobi weight
\[\w(x;t)=x^\a(1-x)^\b(t-x)^\gamma,\qquad x\in[0,1],\]
with $\a,\b,\gamma>0$ and $t\in\R^+$ parameters \cite{refBCM12,refCM12,refDZ10,refMagnus95}.

Recurrence coefficients for orthogonal polynomials with respect to discontinuous weights which involve the Heaviside function $\mathcal{H}(x)$ have also been expressed in terms of solutions of \p\ equations \cite{refBC09,refCZ10,refFO10,refFW03}, while
recurrence coefficients for orthogonal polynomials with respect to discrete weights have been expressed in terms of solutions of \p\ equations
\cite{refBFSvAZ,refBFvA,refPAC13,refFvA11,refFvA13}.

This paper is organized as follows: 
in \S\ref{sec:op}, we review some properties of orthogonal polynomials;
in \S\ref{sec:PIV}, we review some properties of the fourth \peq\ (\ref{eq:PIV}), including its Hamiltonian structure \S\ref{ssec:Ham}, \bk\ and Schlesinger transformations \S\ref{sec:bts} and special function solutions \S\ref{sec:sf};
in \S\ref{sec:sclag} we express the coefficients which arise in the three-term recurrence relation
associated with orthogonal polynomials for the semi-classical Laguerre weight (\ref{eq:scLag1}) in terms of Wronskians that arise in the description of 
special function solutions of \PIV\ (\ref{eq:PIV}); 
in \S\ref{sec:asymp} we derive asymptotic expansions for the recurrence coefficients;
in \S\ref{sec:scherm} we discuss orthogonal polynomials with respect to the semi-classical Hermite weight
\[\w(x;t)=|x|^{\la}\exp(-x^2+tx),\qquad x,t\in\R,\quad \la>-1,\]
which is an extension of the semi-classical Laguerre weight (\ref{eq:scLag1}) to the whole real line, and show that the recurrence coefficients are also expressed in terms of Wronskians that arise in the description of special function solutions of \PIV\ (\ref{eq:PIV}); 
and in \S\ref{sec:dis} we discuss our results.

\section{\label{sec:op}Orthogonal polynomials}
Let $P_n(x)$, $n\in\N$, be the monic orthogonal polynomial of degree $n$ in $x$ with respect to a positive weight $\w(x)$ 
on $(a,b)$, a finite or infinite interval in $\R$, such that
\[\intS P_m(x)P_n(x)\,\w(x)\,\d x = h_n\delta_{m,n},\qquad h_n>0,\]
where $\delta_{m,n}$ denotes the Kronekar delta. 
One of the most important properties of orthogonal polynomials  is that they satisfy a three-term recurrence relationship of the form
\begin{equation}\label{eq:rr}
xP_n(x)=P_{n+1}(x)+\a_nP_n(x)+\b_nP_{n-1}(x),
\end{equation}
where the coefficients $\a_n$ and $\b_n$ are given by the integrals
\[
\a_n = \frac1{h_n}\intS xP_n^2(x)\,\w(x)\,\d x,\qquad \b_n = \frac1{h_{n-1}}\intS xP_{n-1}(x)P_n(x)\,\w(x)\,\d x,
\]
with $P_{-1}(x)=0$ and $P_{0}(x)=1$. These coefficients in the three-term recurrence relationship can also be expressed in terms of determinants whose entries are given in terms of the moments associated with the weight $\w(x)$. Specifically, the coefficients $\a_n$ and $\b_n$ in the recurrence relation (\ref{eq:rr}) are given by
\begin{equation}\label{eq:anbn}
\a_n = \frac{\widetilde{\Delta}_{n+1}}{\Delta_{n+1}}-\frac{\widetilde{\Delta}_n}{\Delta_{n}},\qquad
\b_n = \frac{\Delta_{n+1}\Delta_{n-1}}{\Delta_{n}^2},\end{equation}
where $\Delta_n$ is the Hankel determinant 
\begin{subequations}\label{eq:dets}
\begin{equation}\label{eq:detsDn}
\Delta_n=\det\Big[\mu_{j+k}\Big]_{j,k=0}^{n-1}=\left|\begin{matrix} \mu_0 & \mu_1 & \ldots & \mu_{n-1}\\
\mu_1 & \mu_2 & \ldots & \mu_{n}\\
\vdots & \vdots & \ddots & \vdots \\
\mu_{n-1} & \mu_{n} & \ldots & \mu_{2n-2}\end{matrix}\right|,\qquad n\geq1,\end{equation}
with $\Delta_0=1$, $\Delta_{-1}=0$, 
and $\widetilde{\Delta}_n$ is the determinant 
\begin{equation}\label{eq:detsDDn}\qquad 
\widetilde{\Delta}_n=\left|\begin{matrix} \mu_0 & \mu_1 & \ldots & \mu_{n-2} & \mu_n\\
\mu_1 & \mu_2 & \ldots & \mu_{n-1}& \mu_{n+1}\\
\vdots & \vdots & \ddots & \vdots & \vdots \\
\mu_{n-1} & \mu_{n} & \ldots & \mu_{2n-3}& \mu_{2n-1}\end{matrix}\right|,\qquad n\geq1,
\end{equation}\end{subequations}
with $\widetilde{\Delta}_0=0$ and $\mu_k$, the $k$th moment, is given by the integral
\begin{equation}\label{eq:moment}
\mu_k=\intS  x^k\w(x)\,\d x.\end{equation}
We remark that the Hankel determinant $\Delta_n$ (\ref{eq:detsDn}) also has the integral representation
\begin{equation}\label{eq:detsDnint}
\Delta_n=\frac{1}{n!}\intS 
\dotsn\intS\prod_{\ell=1}^n\w(x_\ell)\prod_{1\leq j<k\leq n}(x_j-x_k)^2\,\d x_1\,\ldots\,\d x_n,\qquad n\geq1.
\end{equation}
The monic polynomial $P_n(x)$ can be uniquely expressed as the determinant
\[P_n(x)=\frac1{\Delta_n}\left|\begin{matrix} \mu_0 & \mu_1 & \ldots & \mu_{n}\\
\mu_1 & \mu_2 & \ldots & \mu_{n+1}\\
\vdots & \vdots & \ddots & \vdots \\
\mu_{n-1} & \mu_{n} & \ldots & \mu_{2n-1}\\
1 & x & \ldots &x^n\end{matrix}\right|,\]
and the normalisation constants as
\begin{equation}\label{def:norm}h_n=\frac{\Delta_{n+1}}{\Delta_n},\qquad h_0=\Delta_1=\mu_0.
\end{equation}
For further information about orthogonal polynomials see, for example \cite{refChihara,refIsmail,refSzego}.

Now suppose that the weight has the form 
\begin{equation}\label{scweight}w(x;t)=\w_0(x)\exp(xt),\qquad x\in[a,b],\end{equation}
where $t$ is a parameter, with finite moments for all $t\in\R$, which is the case for the semi-classical Laguerre weight (\ref{eq:scLag1}). 
If the weight has the form (\ref{scweight}), which depends on the parameter $t$, then the orthogonal polynomials $P_n(x)$, 
the recurrence coefficients $\a_n$, $\b_n$ given by (\ref{eq:anbn}),
the determinants $\Delta_n$, $\widetilde{\Delta}_n$ given by (\ref{eq:dets}) and 
the moments $\mu_k$ given by (\ref{eq:moment}) are now functions of $t$. 
Specifically, in this case then
\[\begin{split}
\mu_k &
= \intS x^k \w_0(x)\exp(xt)\,\d x 
=\deriv[k]{}{t}\left( \intS\w_0(x)\exp(xt)\,\d x\right)
=\deriv[k]{\mu_0}{t}.
\end{split}\]
Further, the recurrence relation has the form
\begin{equation}\label{eq:scrr}xP_n(x;t)=P_{n+1}(x;t)+\a_n(t)P_n(x;t)+\b_n(t)P_{n-1}(x;t),\end{equation}
where we have explicitly indicated that the coefficients $\a_n(t)$ and $\b_n(t)$ depend on $t$.

\begin{theorem}{\label{thm:2.1}If the weight has the form (\ref{scweight}), then the determinants $\Delta_n(t)$ and $\widetilde{\Delta}_n(t)$ given by (\ref{eq:dets}) can be written as 
\begin{align}\label{scHank}
\Delta_n(t)&=\W\left(\mu_0,\deriv{\mu_0}t,\ldots,\deriv[n-1]{\mu_0}t\right),\qquad 
\widetilde{\Delta}_n(t)=\deriv{\Delta_n}{t},
\end{align}
where $\W(\ph_1,\ph_2,\ldots,\ph_n)$ is the Wronskian given by
$$\W(\ph_1,\ph_2,\ldots,\ph_n)=\left|\begin{matrix} 
\ph_1 & \ph_2 & \ldots & \ph_n\\
\ph_1^{(1)} & \ph_2^{(1)} & \ldots & \ph_n^{(1)}\\
\vdots & \vdots & \ddots & \vdots \\
\ph_1^{(n-1)} & \ph_2^{(n-1)} & \ldots & \ph_n^{(n-1)}
\end{matrix}\right|,\qquad \ph_j^{(k)}=\deriv[k]{\ph_j}{t}.$$
}\end{theorem}

\begin{proof}{Since $\ds\mu_k=\deriv[k]{\mu_0}{t}$,
the determinant $\Delta_n(t)$ can be written in the form
\[
\Delta_n(t)=
\left|\begin{matrix} \mu_0 & \mu_1 & \ldots & \mu_{n-1}\\
\mu_1 & \mu_2 & \ldots & \mu_{n}\\
\vdots & \vdots & \ddots & \vdots \\
\mu_{n-1} & \mu_{n} & \ldots & \mu_{2n-2}\end{matrix}\right|=
\W\left(\mu_0,\deriv{\mu_0}t,\ldots,\deriv[n-1]{\mu_0}t\right),
\]
as required, and the determinant $\widetilde{\Delta}_n(t)$, can be written in the form
\[\begin{split}
\widetilde{\Delta}_n(t)&=\left|\begin{matrix} \mu_0 & \mu_1 & \ldots & \mu_{n-2}& \mu_{n}\\
\mu_1 & \mu_2 & \ldots & \mu_{n-1}& \mu_{n+1}\\
\vdots & \vdots & \ddots & \vdots & \vdots \\
\mu_{n-1} & \mu_{n} & \ldots & \mu_{2n-3}& \mu_{2n-1}\end{matrix}\right| 
=\W\left(\mu_0,\deriv{\mu_0}t,\ldots,\deriv[n-2]{\mu_0}t,\deriv[n]{\mu_0}t\right)\\
&=\deriv{}{t}\W\left(\mu_0,\deriv{\mu_0}t,\ldots,\deriv[n-1]{\mu_0}t\right)
=\deriv{\Delta_n}{t},
\end{split}\]
as required.}\end{proof}

The Hankel determinant $\Delta_n(t)$ satisfies the Toda equation, as shown in the following theorem.

\begin{theorem}{\label{thm:toda}The Hankel determinant $\Delta_n(t)$ given by (\ref{scHank}) satisfies the Toda equation
\begin{equation}\label{eq:toda1}
\deriv[2]{}{t}\ln\Delta_n(t)=\frac{\Delta_{n-1}(t)\Delta_{n+1}(t)}{\Delta_{n}^2(t)}.
\end{equation}
}\end{theorem}
\begin{proof}{See, for example, Nakamira and Zhedanov \cite[Proposition 1]{refNZ}; also \cite{refCI97,refPSZ,refSogo}.}\end{proof}

Using Theorems \ref{thm:2.1} and \ref{thm:toda} we can express the recurrence coefficients $\a_n(t)$ and $\b_n(t)$ in terms of derivatives of the Hankel determinant $\Delta_n(t)$ and so obtain explicit expressions for these coefficients.

\begin{theorem}{\label{thm:anbn}The coefficients $\a_n(t)$ and $\b_n(t)$ in the recurrence relation (\ref{eq:scrr})
associated with monic polynomials orthogonal with respect to a weight of the form (\ref{scweight}) are given by
\[\label{def:anbn} \a_n(t) =\deriv{}{t}\ln \frac{\Delta_{n+1}(t)}{\Delta_{n}(t)},\qquad 
\b_n(t) =
\deriv[2]{}{t}\ln\Delta_n(t) ,\]
with $\Delta_n(t)$ is the Hankel determinant given by (\ref{scHank}).
}\end{theorem}

\begin{proof}{By definition the coefficients $\a_n(t)$ and $\b_n(t)$ in the recurrence relation (\ref{eq:scrr}) are given by
\[\a_n(t) = \frac{\widetilde{\Delta}_{n+1}(t)}{\Delta_{n+1}(t)}-\frac{\widetilde{\Delta}_n(t)}{\Delta_{n}(t)},\qquad
\b_n(t) = \frac{\Delta_{n-1}(t)\,\Delta_{n+1}(t)}{\Delta_{n}^2(t)},\]
where the determinants $\Delta_n$ and $\widetilde{\Delta}_n$ are given by (\ref{eq:dets}).
Hence from (\ref{scHank})
\[\begin{split}\a_n(t)&= \frac{\widetilde{\Delta}_{n+1}(t)}{\Delta_{n+1}(t)}-\frac{\widetilde{\Delta}_n(t)}{\Delta_{n}(t)} 
= \frac1{\Delta_{n+1}}\deriv{\Delta_{n+1}}{t}-\frac1{\Delta_n}\deriv{\Delta_n}{t},\end{split}\]
and so
\begin{equation*}\label{eq:ann}
\a_n(t)=\deriv{}{t}\ln \frac{\Delta_{n+1}(t)}{\Delta_{n}(t)},\end{equation*}
as required. By definition 
$$\b_n(t)=\frac{\Delta_{n-1}(t)\Delta_{n+1}(t)}{\Delta_{n}^2(t)},$$ and so from Theorem \ref{thm:toda} we have
\begin{equation*}\label{eq:bnn}
\b_n(t)=\deriv[2]{}{t}\ln\Delta_{n}(t),
\end{equation*}
as required. See also Chen, Ismail and van Assche  \cite{refCIvA98} who also discuss applications to random matrices.
}\end{proof}

Equivalently the recurrence coefficients $\a_n(t)$ and $\b_n(t)$ can be expressed in terms of $h_n(t)$ given by (\ref{def:norm}).
\begin{lemma}{\label{thm:anbn2}The coefficients $\a_n(t)$ and $\b_n(t)$ in the recurrence relation (\ref{eq:scrr})
associated with monic polynomials orthogonal with respect to a weight of the form (\ref{scweight}) are given by
\begin{equation*}\label{def:anbn2} \a_n(t) =\deriv{}{t}\ln h_{n}(t),\qquad 
\b_n(t) =\frac{h_{n+1}(t)}{h_{n}(t)},\end{equation*}
where $h_n(t)$ is given by (\ref{def:norm}).
}\end{lemma}

\begin{proof}{See Chen and Ismail \cite{refCI97}.}\end{proof}

Additionally the coefficients $\a_n(t)$ and $\b_n(t)$ in the recurrence relation (\ref{eq:scrr}) satisfy a Toda system.

\begin{theorem}{\label{thm:todasys}
The coefficients $\a_n(t)$ and $\b_n(t)$ in the recurrence relation (\ref{eq:scrr}) associated with a weight of the form (\ref{scweight}) satisfy the Toda system
\begin{equation}\label{eq:toda}\deriv{\a_n}t=\b_{n+1}-\b_n,\qquad\deriv{\b_n}t=\b_n(\a_n-\a_{n-1}).\end{equation}
}\end{theorem}
\begin{proof}{See Chen and Ismail \cite{refCI97}, Ismail \cite[\S2.8, p.\ 41]{refIsmail} and Moser \cite{refMoser}; see also \cite{refBFvA} for further details and a direct proof in the case of a semi-classical weight of the form (\ref{scweight}).}\end{proof}

Suppose  $P_n(x)$, for $n\in\N$,  
is a sequence of \textit{classical} orthogonal polynomials (such as Hermite, Laguerre and Jacobi polynomials), then $P_n(x)$ is a solution of a second-order \ode\ of the form
\begin{equation}\label{eq:Pn}
\sigma(x)\deriv[2]{P_n}{x}+\tau(x)\deriv{P_n}{x}=\lambda_nP_n,
\end{equation}
where $\sigma(x)$ is a monic polynomial with deg$(\sigma)\leq2$, $\tau(x)$ is a polynomial with deg$(\tau)=1$, and $\lambda_n$ is a real number which depends on the degree of the polynomial solution, see Bochner \cite{refBochner}. Equivalently, the weights of classical orthogonal polynomials satisfy a first-order \ode, the \textit{Pearson equation}
\begin{equation}\label{eq:Pearson}
\deriv{}{x}[\sigma(x)\w(x)]=\tau(x)\w(x),
\end{equation}
with $\sigma(x)$ and $\tau(x)$ the same polynomials as in (\ref{eq:Pn}), see, for example \cite{refAlNod,refBochner,refChihara}. However for \textit{semi-classical} orthogonal polynomials, the weight function $\w(x)$ satisfies the Pearson equation (\ref{eq:Pearson}) with either deg$(\sigma)>2$ or deg$(\tau)>1$, see, for example \cite{refHvR,refMaroni}. For example, the Pearson equation (\ref{eq:Pearson}) is satisfied for the weight (\ref{eq:scLag1}) with
\[\sigma(x)=x,\qquad\tau(x)=-2x^2+tx+\la+1,\]
and so the weight (\ref{eq:scLag1}) is indeed a semi-classical weight function. 
Filipuk, van Assche and Zhang \cite{refFvAZ} comment that \begin{quote}
``\textit{We note that for classical orthogonal polynomials (Hermite, Laguerre, Jacobi) one
knows these recurrence coefficients explicitly in contrast to non-classical weights}".
\end{quote}

In \S\ref{sec:sclag} we show that, in the case of the semi-classical Laguerre weight (\ref{eq:scLag1}), the determinants $\Delta_n(t)$ and $\widetilde{\Delta}_n(t)$ can be explicitly written as Wronskians which arise in the description of special function solutions of \PIV\ (\ref{eq:PIV}) that are expressed in terms of parabolic cylinder functions $\WhitD{\nu}(z)$ when $\la\not\in\Z$, or error functions $\erf(z)$ when $\la=n\in\Z$. Consequently the recurrence coefficients $\a_n(t)$ and $\b_n(t)$ (\ref{eq:anbn}) associated with orthogonal polynomials for the semi-classical Laguerre weight (\ref{eq:scLag1}) can also be explicitly written in terms of these Wronskians.

\section{Properties of the fourth \p\ equation}\label{sec:PIV}
The six \peqs\ (\PI--\PVI) were first discovered by \p, Gambier and their colleagues in an investigation of which second order \odes\ of the form 
\begin{equation} \label{eq:PT.INT.gen-ode} \qzz=F\left(\qz,q,z\right),  \end{equation} 
where $F$ is rational in $\d q/\d z$ and $q$ and analytic in $z$, have the property that their solutions have no movable branch points. 
They showed that there were fifty canonical equations of the form (\ref{eq:PT.INT.gen-ode}) with this property, now known as the \textit{\p\ property}. Further \p, Gambier and their colleagues showed that of these fifty equations, forty-four can be reduced to linear equations, solved in terms of elliptic functions, or are reducible to one of six new nonlinear ordinary differential equations that define new transcendental functions, see Ince \cite{refInce}. 
The \p\ equations can be thought of as nonlinear analogues of the classical special functions \cite{refPAC05review,refFIKN,refGLS02,refIKSY,refUmemura98}, and arise in a wide variety of applications, for example random matrices, cf.~\cite{refForrester,refOsipovKanz}.

\subsection{Hamiltonian structure}\label{ssec:Ham} Each of the
\peqs\ \PI--\PVI\ can be written as a Hamiltonian system 
\begin{equation}\label{sec:PT.HM.DE1}
\frac{\d q}{\d z}=\pderiv{\mathcal{H}_{\rm J}}{p},\qquad 
\frac{\d p}{\d z}=-\pderiv{\mathcal{H}_{\rm J}}{q},
\end{equation}
for a suitable Hamiltonian function $\mathcal{H}_{\rm J}(q,p,z)$\ \cite{refJMi,refOkamoto80a,refOkamotoPIIPIV}. 
The function $\sigma(z)\equiv\mathcal{H}_{\rm J}(q,p,z)$ satisfies a second-order, second-degree ordinary differential
equation, whose solution is expressible in terms of the solution of the associated \peq\ \cite{refJMi,refOkamoto80b,refOkamotoPIIPIV}.

The Hamiltonian associated with \PIV\ (\ref{eq:PIV}) is
\begin{equation}\label{eq:PT.HM.DE41}
\mathcal{H}_{\rm IV}(q,p,z;\th_0,\th_\infty) = 2q p^2-(q^2+2zq+2\th_0)p +\th_{\infty} q,
\end{equation} with $\th_0$ and $\th_{\infty}$ parameters \cite{refJMi,refOkamoto80a,refOkamoto80b,refOkamotoPIIPIV}, and so from (\ref{sec:PT.HM.DE1})
\begin{subequations}\label{eq:PT.HM.DE423}
\begin{align}\label{eq:PT.HM.DE42} \deriv{q}z&=4qp-q^2-2zq-2\th_0,\\
\label{eq:PT.HM.DE43} \deriv{p}z&=-2p^2+2qp+2zp-\th_{\infty}.
\end{align}\end{subequations}
Solving (\ref{eq:PT.HM.DE423}a) for $p$ and substituting in (\ref{eq:PT.HM.DE423}b) yields
\[
\deriv[2]{q}{z} = \frac{1}{2q}\left(\deriv{q}z\right)^2+\tfrac32 q^3+4zq^2+2(z^2+\th_0-2\th_\infty-1)q-\frac{2\th_0^2}{q},
\]
which is \PIV\ (\ref{eq:PIV}) with  $A=1-\th_0+2\th_{\infty}$ and $B=-2\th_0^2$. 
Analogously, solving (\ref{eq:PT.HM.DE423}b) for $q$ and substituting in (\ref{eq:PT.HM.DE423}a) yields
\[\deriv[2]{p}z = \frac{1}{2p}\left(\deriv{p}z\right)^2+6p^3-8zp^2+2(z^2-2\th_0+\th_\infty+1)p-\frac{\th_\infty^2}{2p}.
\] Then letting $p=-\tfrac12w$ yields
\PIV\ (\ref{eq:PIV}) with  $A=-1+2\th_0-\th_{\infty}$ and $B=-2\th_{\infty}^2$.

An important property of the Hamiltonian, which is very useful in applications, is that it satisfies a second-order, second-degree \ode.

\begin{theorem}\label{thm:2.2}Consider the function 
\begin{equation*}\label{eq:PT.HM.DE41a}
\sigma(z;\th_0,\th_\infty) = 2q p^2-(q^2+2zq+2\th_0)p +\th_{\infty} q,
\end{equation*} where $q$ and $p$ satisfy the system (\ref{eq:PT.HM.DE423}), then $\sigma$
satisfies the second-order, second-degree \ode\
\begin{equation}\label{eq:PT.HM.DE44}
\left(\deriv[2]{\sigma}z\right)^{2} - 4\left(z\deriv{\sigma}z-\sigma\right)^{2} +4\deriv{\sigma}z\left(\deriv{\sigma}z+2\th_0\right)\left(\deriv{\sigma}z+2\th_{\infty}\right)=0.
\end{equation} 
Conversely, if $\sigma$ is a solution of (\ref{eq:PT.HM.DE44}), then solutions of the Hamiltonian system (\ref{eq:PT.HM.DE423}) are given by
\begin{equation*}\label{eq:PT.HM.DE45}
q=\frac{\ds\sigma''-2z\sigma'+2\sigma}{\ds2\left(\sigma'+2\th_{\infty}\right)},\qquad
p=\frac{\ds\sigma''+2z\sigma'-2\sigma}{\ds4\left(\sigma'+2\th_0\right)},\qquad '=\deriv{}{z}.
\end{equation*}\end{theorem}

\begin{proof}See Jimbo and Miwa \cite{refJMi} and Okamoto \cite{refOkamoto80a,refOkamoto80b,refOkamotoPIIPIV}.\end{proof}

\begin{remarks}{\rm\ \phantom{x}\ 
\begin{enumerate} 
\item Equation (\ref{eq:PT.HM.DE44}), which is often known as \sPIV\ (or the \PIV\ $\sigma$-equation), is equivalent to equation SD-I.c in the classification of second order, second-degree \odes\ with the \p\ property by Cosgrove and Scoufis \cite{refCS}, an equation first derived and solved by Chazy \cite{refChazy11} and subsequently by Bureau \cite{refBureau64,refBureau72} by expressing the solution in terms of solutions of \PIV. 
\item Theorem \ref{thm:2.2} shows that solutions of equation (\ref{eq:PT.HM.DE44}) 
are in a one-to-one correspondence with solutions of the Hamiltonian system (\ref{eq:PT.HM.DE423}), and so are in a one-to-one correspondence with solutions of \PIV\ (\ref{eq:PIV}).
\item\Eref{eq:PT.HM.DE44}\ also arises in various applications, for example random matrix theory \cite{refFW01,refFW03,refKanzieper,refTW94c}. 
\end{enumerate}}\end{remarks}

\subsection{\bk\ and Schlesinger transformations}\label{sec:bts}
The \peqs\ \PII--\PVI\ possess \emph{\bts}\ which relate one solution to another solution either of the same equation, with different values of the parameters, or another equation (see \cite{refPAC05review,refFA82,refGLS02} and the references therein). An important application of the \bts\ is that they generate hierarchies of classical solutions of the \peqs, which are discussed in
\sref{sec:sf}.

\bts\ for \PIV\ (\ref{eq:PIV}) are given as follows.

{\begin{theorem}\label{thm:bts}Let $q_0=w(z;A_0,B_0)$ and $q_j^{\pm}= w(z;A_j^{\pm},B_j^{\pm})$,
$j=1,2,3,4$ be solutions of \PIV\ (\ref{eq:PIV}) with
\begin{align*}
A_1^{\pm} &= \tfrac14(2-2A_0 \pm 3\sqrt{-2B_0}),& B_1^{\pm} &= -\tfrac12(1+A_0 \pm \tfrac12\sqrt{-2B_0})^{2},\\
A_2^{\pm}&= -\tfrac14(2+2A_0 \pm 3\sqrt{-2B_0}),& B_2^{\pm} &= -\tfrac12(1-A_0 \pm \tfrac12\sqrt{-2B_0})^{2},\\
A_3^{\pm}&=\tfrac32-\tfrac12A_0\mp\tfrac34\sqrt{-2B_0},& B_3^{\pm}&=-\tfrac12(1- A_0\pm\tfrac12\sqrt{-2B_0})^{2},\\
A_4^{\pm}&=-\tfrac32-\tfrac12A_0\mp\tfrac34\sqrt{-2B_0},&
B_4^{\pm}&=-\tfrac12(-1-A_0 \pm\tfrac12\sqrt{-2B_0})^{2}.
\end{align*}
Then
\begin{subequations}\label{eq:PIVbts}\begin{align}
\mathcal{T}_1^{\pm}:\qquad q_1^{\pm} &=\frac{q_0' -q_0^2 -2zq_0\mp\sqrt{-2B_0}}{2q_0},\\
\mathcal{T}_2^{\pm}:\qquad q_2^{\pm} &=-\,\frac{q_0' +q_0^2+2zq_0\mp\sqrt{-2B_0}}{2q_0},\\
\mathcal{T}_3^{\pm}:\qquad q_3^{\pm} &=q_0+\frac{2\left(1-A_0\mp\tfrac12\sqrt{-2B_0}\,\right)q_0}{q_0' \pm\sqrt{-2B_0}+2zq_0 +q_0^2},\\
\mathcal{T}_4^{\pm}:\qquad q_4^{\pm} &= q_0+\frac{2\left(1+A_0\pm\tfrac12\sqrt{-2B_0}\,\right)q_0} {q_0'\mp\sqrt{-2B_0}-2zq_0 -q_0^2},
\end{align}\end{subequations} valid when the denominators are non-zero,
and where the upper signs or the lower signs are taken throughout each
transformation.\end{theorem}

\begin{proof}See Gromak \cite{refGromak78,refGromak87} and Lukashevich \cite{refLuk67a}; 
also \cite{refBCH95,refGLS02,refMurata85}.
\end{proof}}

A class of \bts\ for the \peqs\ is generated by so-called \textit{Schlesinger transformations} of the associated isomonodromy problems.
Fokas, Mugan and Ablowitz \cite{refFMA}, deduced the following Schlesinger transformations
$\mathcal{R}_1$--$\mathcal{R}_4$ for \PIV.
\begin{subequations}\label{eq:Schles}\begin{align} 
\mathcal{R}_1:\quad q_1(z;A_1,B_1)
&= \frac{\left(q'+\sqrt{-2 B}\right)^2+\left(4A+4-2\sqrt{-2 B}\right)q^2-q^2(q+2z)^2}{2q\left(q^2+2zq-q'-\sqrt{-2 B}\right)},\\
\mathcal{R}_2:\quad q_2(z;A_2,B_2)
&= \frac{\left(q'-\sqrt{-2 B}\right)^2+\left(4A-4-2\sqrt{-2 B}\right)q^2-q^2(q+2z)^2}{2q\left(q^2+2zq+q'-\sqrt{-2 B}\right)},\\
\mathcal{R}_3:\quad q_3(z;A_3,B_3) 
&= \frac{\left(q'-\sqrt{-2 B}\right)^2-\left(4A+4+2\sqrt{-2 B}\right)q^2-q^2(q+2z)^2}{2q\left(q^2+2zq-q'+\sqrt{-2 B}\right)},\\
\mathcal{R}_4:\quad q_4(z;A_4,B_4)
&= \frac{\left(q'+\sqrt{-2 B}\right)^2+\left(4A-4+2\sqrt{-2 B}\right)q^2-q^2(q+2z)^2}{2q\left(q^2+2zq+q'+\sqrt{-2 B}\right)},
\end{align} where $q\equiv q(z;A,B)$ and
\begin{align}
(A_1,B_1)&=\left(A+1,-\tfrac12 \big(2-\sqrt{-2B}\big)^2\right), && (A_2,B_2)=\left(A-1,-\tfrac12 \big(2+\sqrt{-2B}\big)^2\right),\\
(A_3,B_3)&=\left(A+1,-\tfrac12 \big(2+\sqrt{-2B}\big)^2\right), && (A_4,B_4)=\left(A-1,-\tfrac12 \big(2-\sqrt{-2B}\big)^2\right).
\end{align}\end{subequations} 
Fokas, Mugan and Ablowitz \cite{refFMA} also defined
the composite transformations $\mathcal{R}_5=\mathcal{R}_1\mathcal{R}_3$
and
$\mathcal{R}_7=\mathcal{R}_2\mathcal{R}_4$ given by
\begin{subequations}\label{STR567}\begin{align}\label{STR5}  
\mathcal{R}_5:\quad q_5(z;A_5,B_5)
&=\frac{\left(q'-q^2-2zq\right)^2+2B}{2q\left\{ q'-q^2-2zq+2\left(A+1 \right)\right\}},\\
\mathcal{R}_7:\quad q_7(z;A_7,B_7)
&=-\,\frac{\left(q'+q^2+2zq\right)^2+2B}{2q\left\{q'+q^2+2zq-2\left(A-1 \right)\right\}},\label{STR7}  
\end{align}
respectively, where 
\begin{equation}
(A_5,B_5)=(A+2,B),\qquad 
(A_7,B_7)=(A-2,B).
\end{equation}
\end{subequations}
We remark that $\mathcal{R}_5$ and $\mathcal{R}_7$ are the transformations
$\mathcal{T}_+$ and $\mathcal{T}_-$, respectively, given by Murata \cite{refMurata85}.

\subsection{Special function solutions}\label{sec:sf} 
The \peqs\ \PII--\PVI\ possess hierarchies of solutions expressible in terms of classical special functions, for special values of the parameters through an associated Riccati equation, 
\begin{equation}
\label{eq:PT.SF.eq20} \deriv{q}{z}=f_2(z)q^2+f_1(z)q+f_0(z),
\end{equation} 
where $f_2(z)$, $f_1(z)$ and $f_0(z)$ are rational functions. Hierarchies of solutions, which are often referred to as ``one-parameter solutions" (since they have one arbitrary constant), are generated from ``seed solutions'' derived from the Riccati equation using the \bts\ given in \sref{sec:bts}. Furthermore, as for the rational solutions, these special function solutions are often expressed in the form of determinants. 

Solutions of \PII--\PVI\ are expressed in terms of special functions as follows (see \cite{refPAC05review,refGLS02,refMasuda04}, and the references therein):
for \PII\ in terms of Airy functions $\Ai(z)$ and $\Bi(z)$;
for \PIII\ in terms of Bessel functions $\BesselJ{\nu}(z)$ and $\BesselY{\nu}(z)$; 
for \PIV\ in terms of parabolic cylinder functions functions $\WhitD{\nu}(z)$; 
for \PV\ in terms of confluent hypergeometric functions $\HyperpFq11(a;c;z)$ (equivalently Kummer functions $\KummerM(a,b,z)$ and $\KummerU(a,b,z)$ or Whittaker functions $\WhitM{\k}{\mu}(z)$ and $\WhitW{\k}{\mu}(z)$); 
and  for \PVI\ in terms of hypergeometric functions $\HyperpFq21(a,b;c;z)$. 
Some classical orthogonal polynomials arise as particular cases of these special function solutions and thus yield rational solutions of the associated \peqs:  
for \PIII\ and \PV\ in terms of associated Laguerre polynomials $\LaguerreL{m}{k}(z)$; 
for \PIV\ in terms of Hermite polynomials $\HermiteH{n}(z)$; and 
for \PVI\ in terms of Jacobi polynomials $\JacobiP{\a}{\b}{n}(z)$. 

Special function solutions of \PIV\ (\ref{eq:PIV}) are expressed in in terms of parabolic cylinder functions. 

\begin{theorem}\label{thm:p4sf}\PIV\ (\ref{eq:PIV}) has solutions expressible in terms of parabolic cylinder 
functions if and only if either
\begin{equation}\label{eq:PT.SF.eq40a} 
B=-2(2n+1+\ep A)^2,
\end{equation} or
\begin{equation} \label{eq:PT.SF.eq40b} 
B=-2n^2,
\end{equation} with $n\in\Integer$ and $\ep=\pm1$.
\end{theorem} 

\begin{proof}See \cite{refGambier09,refGromak87,refGLS02,refGroLuk82,refLuk65,refLuk67a}.\end{proof}

For \PIV\ (\ref{eq:PIV}) the associated Riccati equation is
\begin{equation}\label{eq:PT.SF.eq41}
\deriv{q}{z}=\ep(q^2+2zq)+2\nu,\qquad\ep^2=1,\end{equation} 
with \PIV\ parameters $A=-\ep(\nu+1)$ and $B=-2\nu^2$.
Letting $\ds w(z)=\deriv{}{z}\ln\ph_{\nu}(z)$ in (\ref{eq:PT.SF.eq41}) yields 
\begin{equation}\label{eq:PT.SF.eq41a}
\deriv[2]{\ph_{\nu}}z-2\ep z\deriv{\ph_{\nu}}{z}+2\ep\nu\ph_{\nu}=0.\end{equation}
The solution of this equation depends on whether $\nu\in\Z$ or $\nu\not\in\Z$, which we now summarize.

\begin{itemize}
\item[(i)] If $\nu\not\in\Z$ then equation (\ref{eq:PT.SF.eq41a}) has solutions
\begin{equation}\label{eq:PT.SF.eq44nu}
\ph_\nu(z;\ep)=\begin{cases}
\left\{C_1\WhitD{\nu}(\sqrt{2}\,z)+C_2\WhitD{\nu}(-\sqrt{2}\,z)\right\} \exp\left(\tfrac12 z^2\right), &\mbox{\rm if}\quad \ep=1,\\
\left\{C_1\WhitD{-\nu-1}(\sqrt{2}\,z)+C_2\WhitD{-\nu-1}(-\sqrt{2}\,z)\right\}\exp\left(-\tfrac12 z^2\right), &\mbox{\rm if}\quad \ep=-1,
\end{cases}\end{equation} 
with $C_1$ and $C_2$ arbitrary constants, where
$\WhitD{\nu}(\zeta)$ is the \textit{parabolic cylinder function} which satisfies
\begin{equation}\label{eq:pcf}
\deriv[2]{\WhitD{\nu}}{\zeta}=(\tfrac14 \zeta ^{2}-\nu-\tfrac12)\WhitD{\nu},\end{equation}
and the boundary condition
\[\WhitD{\nu}(\zeta)\sim\zeta^\nu\exp\left(-\tfrac14\zeta ^2\right),\qquad\mbox{as}\quad\zeta\to+\infty.\]

\item[(ii)] If $\nu=0$ then equation (\ref{eq:PT.SF.eq41a}) has the solutions
\begin{equation*}\label{eq:PT.SF.eq440}
\ph_0(z;\ep)=\begin{cases}
C_1 + C_2\erfi(z),\quad &\mbox{\rm if}\quad \ep=1,\\
C_1 + C_2\erfc(z),  &\mbox{\rm if}\quad \ep=-1,
\end{cases}\end{equation*} 
with $C_1$ and $C_2$ arbitrary constants, where $\erfc(z)$ is the \textit{complementary error function} 
and $\erfi(z)$ is the \textit{imaginary error function}, respectively defined by
\begin{equation}\label{def:erfc}
\erfc (z)=\frac{2}{\sqrt{\pi}} \int_z^\infty \exp (-t^2) \,\d t,\qquad
\erfi (z)=\frac{2}{\sqrt{\pi}} \int_0^z \exp (t^2) \,\d t. 
\end{equation}

\item[(iii)]  If $\nu=m$, for $m\geq1$, then equation (\ref{eq:PT.SF.eq41a}) has the solutions
\begin{equation*}\label{eq:PT.SF.eq44m}
\ph_{m}(z;\ep)=\begin{cases}
\ds C_1H_m(z)+ C_2\exp(z^2)\deriv[m]{}{z}\left\{\erfi(z)\exp(-z^2)\right\}, &\mbox{\rm if}\quad \ep=1,\\[2.5pt]
\ds C_1(-\i)^{m}H_m(\i z)+ C_2\exp(-z^2)\deriv[m]{}{z}\left\{\erfc(z)\exp(z^2)\right\},\quad  &\mbox{\rm if}\quad \ep=-1,
\end{cases}\end{equation*} 
with $C_1$ and $C_2$ arbitrary constants, where $H_m(z)$ is the \textit{Hermite polynomial} defined by
\begin{equation}\label{def:hermite}H_m(z)=(-1)^m\exp(z^2)\deriv[m]{}{z}\exp(-z^2).\end{equation}

\item[(iv)]  If $\nu=-m$, for $m\geq1$, then equation (\ref{eq:PT.SF.eq41a}) has the solutions
\begin{equation*}\label{eq:PT.SF.eq44n}
\ph_{-m}(z;\ep)=\begin{cases}
\ds C_1(-\i)^{m-1}H_{m-1}(\i z)\exp(z^2)+ C_2\deriv[m-1]{}{z}\left\{\erfc(z)\exp(z^2)\right\}, &\mbox{\rm if}\quad \ep=1,\\[2.5pt]
\ds C_1H_{m-1}(z)\exp(-z^2) + C_2\deriv[m-1]{}{z}\left\{\erfi(z)\exp(-z^2)\right\},\quad  &\mbox{\rm if}\quad \ep=-1,
\end{cases}\end{equation*} 
with $C_1$ and $C_2$ arbitrary constants. 
\end{itemize}

If $\ph_{\nu}(z;\ep)$ is a solution of (\ref{eq:PT.SF.eq41a}), then the ``seed solutions" of \PIV\ (\ref{eq:PIV}) are given by
\[\begin{split} 
q\big(z;-\ep(\nu+1),-2\nu^2\big)&=-\ep\deriv{}z\ln\ph_{\nu}(z;\ep),\qquad 
q\big(z;-\ep\nu,-2(\nu+1)^2\big)=-2z+\ep\deriv{}z\ln\ph_{\nu}(z;\ep).
\end{split}\]
Hierarchies of special function solutions can be generated from these solutions using the \bts\ given in \sref{sec:bts}. However there is an alternative approach.

Determinantal representations of special function solutions for \PIV\ (\ref{eq:PIV}) and \sPIV\ (\ref{eq:PT.HM.DE44}) 
are discussed in the following theorem.

\begin{theorem}{\label{PIVtau}Let $\tau_{n,\nu}(z;\ep)$ be given by 
\begin{equation}\label{eq:taudet1}
\tau_{n,\nu}(z;\ep)=\W\left(\ph_{\nu}(z;\ep),\deriv{\ph_{\nu}}{z}(z;\ep),\ldots,\deriv[n-1]{\ph_{\nu}}{z}(z;\ep)\right),\qquad n\geq1,
\end{equation}
with $\tau_{0,\nu}(z;\ep)=1$, where $\ph_{\nu}(z;\ep)$ is a solution of (\ref{eq:PT.SF.eq41a}) and
$\W(\ph_{1},\ph_{2},\ldots,\ph_{n})$ is the Wronskian. 
Then for $n\geq0$, special function solutions of \PIV\ (\ref{eq:PIV}) are given by
\begin{subequations}\label{eq:p4tausols123}
\begin{align} \label{p4tausols1} 
q_{n,\nu}^{[1]}\left(z;A_{n,\nu}^{[1]},B_{n,\nu}^{[1]}\right) 
&=-2z+\ep\deriv{}{z}\ln\frac{\tau_{n+1,\nu}(z;\ep)}{\tau_{n,\nu}(z;\ep)},
&& A_{n,\nu}^{[1]}=\ep(2n-\nu),&& B_{n,\nu}^{[1]}=-2(\nu+1)^2,\\ \label{p4tausols2} 
q_{n,\nu}^{[2]} \left(z;A_{n,\nu}^{[2]},B_{n,\nu}^{[2]}\right) 
&=\ep\deriv{}{z}\ln\frac{\tau_{n,\nu}(z;\ep)}{\tau_{n,\nu+1}(z;\ep)},
&&A_{n,\nu}^{[2]}=\ep(2\nu-n),&& B_{n,\nu}^{[2]}=-2(n+1)^2,\\ \label{p4tausols3} 
q_{n,\nu}^{[3]} \left(z;A_{n,\nu}^{[3]},B_{n,\nu}^{[3]}\right) 
&=\ep\deriv{}{z}\ln\frac{\tau_{n,\nu+1}(z;\ep)}{\tau_{n+1,\nu}(z;\ep)},
&& A_{n,\nu}^{[3]}=-\ep(n+\nu),&& B_{n,\nu}^{[3]}=-2(\nu-n+1)^2,
\end{align}\end{subequations}
and special function solutions of \sPIV\ (\ref{eq:PT.HM.DE44}) are given by
\begin{subequations}\begin{align}\label{eq:PT.HM.DE44a}
\sigma_{n,\nu}^{[1]}(z;\vth_0,\vth_\infty) &= \deriv{}{z}\ln \tau_{n,\nu}(z;\ep),&& \vth^{[1]}_0=\ep(\nu-n+1),&& \vth^{[1]}_{\infty}=-\ep n,\\
\label{eq:PT.HM.DE44b}
\sigma_{n,\nu}^{[2]}(z;\vth_0,\vth_\infty) &= \deriv{}{z}\ln \tau_{n,\nu}(z;\ep)-2\ep nz,&& \vth^{[2]}_{0}=\ep n,&& \vth^{[2]}_\infty=\ep(\nu+1),\\
\sigma_{n,\nu}^{[3]}(z;\vth_0,\vth_\infty) &= \deriv{}{z}\ln \tau_{n,\nu}(z;\ep)+2\ep(\nu-n+1)z,
&& \vth^{[3]}_{0}=-\ep (\nu+1),&& \vth^{[3]}_\infty=-\ep(\nu-n+1),\label{eq:PT.HM.DE44c}
\end{align}\end{subequations} 
}\end{theorem}
\begin{proof}{See Okamoto \cite{refOkamotoPIIPIV}; also Forrester and Witte \cite{refFW01}.}\end{proof}

\section{Semi-classical Laguerre weight}\label{sec:sclag}
In this section we consider monic orthogonal polynomials $P_n(x;t)$, for $n\in\N$, with respect to the {semi-classical Laguerre weight} (\ref{eq:scLag1}),
where these polynomials satisfy the three-term recurrence relation (\ref{eq:scrr}), i.e.
\begin{equation}\label{eq:scLag2}
xP_n(x;t)=P_{n+1}(x;t)+\a_n(t)P_n(x;t)+\b_{n}(t)P_{n-1}(x;t),\end{equation}

Boelen and van Assche \cite[Theorem 1.1]{refBvA} prove the following theorem.
\begin{theorem}{Let $\a_n(t)$ and $\b_n(t)$ be the coefficients in the recurrence relation (\ref{eq:scLag2}) associated with the semi-classical Laguerre weight (\ref{eq:scLag1}). Then the quantities
\begin{equation}\label{eq:xy}x_n=\frac{\sqrt2}{t-2\a_n},\qquad y_n=2\b_n-n-\tfrac12\la,\end{equation}
satisfy the discrete system
\begin{equation}\label{eq:xydisys}
x_{n-1}x_n=\frac{y_n+n+\tfrac12\la}{y_n^2-\tfrac14\la^2},\qquad y_n+y_{n+1}=\frac1{x_n}\left(\frac{t}{\sqrt2}-\frac1{x_n}\right).\end{equation}
}\end{theorem}

Boelen and van Assche \cite{refBvA} also show that the system (\ref{eq:xydisys}) can be obtained from an asymmetric discrete \PIV\ equation by a limiting process. However, from our point of view, it is more convenient to have the discrete system satisfied by $\a_n$ and $\b_n$, which is given in the following Lemma.

\begin{lemma}{The coefficients $\a_n(t)$ and $\b_n(t)$ in the recurrence relation  (\ref{eq:scLag2}) associated with the semi-classical Laguerre weight (\ref{eq:scLag1}) satisfy the discrete system
\begin{subequations}\label{eq:abdisys}\begin{align}
&(2\a_n-t)(2\a_{n-1}-t)=\frac{(2\b_n-n)(2\b_n-n-\la)}{\b_n},\\
&2\b_n+2\b_{n+1}+\a_n(2\a_n-t)=2n+\la+1.
\end{align}\end{subequations}
}\end{lemma}

\begin{proof}{Substituting (\ref{eq:xy}) into (\ref{eq:xydisys}) yields the discrete system (\ref{eq:abdisys}).}\end{proof}

Since the semi-classical Laguerre weight (\ref{eq:scLag1}) has the form $\w_0(x)\exp(xt)$ and the moments are finite for all $t\in\R$, 
with $t$ a parameter, then the coefficients $\a_n(t)$ and $\b_n(t)$ in the recurrence relation (\ref{eq:scLag2}) satisfy the Toda system, recall Theorem \ref{thm:todasys}.

We are now in a position to prove the relationship between the coefficients $\a_n(t)$ and $\b_n(t)$ in the recurrence relation  (\ref{eq:scLag2}) associated with the semi-classical Laguerre weight (\ref{eq:scLag1}) and solutions of \PIV\ (\ref{eq:PIV}).
\begin{theorem}{\label{thm:1.1}
The coefficients $\a_n(t)$ and $\b_n(t)$ in the recurrence relation (\ref{eq:scLag2}) 
associated with the semi-classical Laguerre weight (\ref{eq:scLag1}) are given by
\begin{subequations}\label{eq:scLag3ab}\begin{align}\label{eq:scLag3a}
\a_n(t)&=\tfrac12q_n(z)+\tfrac12t,\\ \label{eq:scLag3b}
\b_n(t)&=-\tfrac18\deriv{q_n}z-\tfrac18 q_n^2(z) -\tfrac14zq_n(z)+\tfrac12n+\tfrac14\la,\end{align}\end{subequations}
with $z=\tfrac12t$, where $q_n(z)$ satisfies
\begin{equation}\label{eq:scLag3}
\deriv[2]{q_n}{z}=\frac1{2q_n}\left(\deriv{q_n}{z}\right)^2+\tfrac32q_n^3+4zq_n^2+2(z^2-2n-\la-1)q_n-\frac{2\la^2}{q_n},
\end{equation}
which is \PIV\ (\ref{eq:PIV}), with parameters
\begin{equation}\label{eq:scLag4}
(A,B)=(2n+\la+1,-2\la^2).
\end{equation}}\end{theorem}

\begin{proof}{Solving the discrete system (\ref{eq:abdisys}) for $\a_{n-1}$ and $\b_{n+1}$ yields
\[\begin{split}\a_{n-1}&=\tfrac12t+\frac{(2\b_n-n)(2\b_n-n-\la)}{2(2\a_n-t)\b_n},\\
\b_{n+1}&=-\b_n-\tfrac12(2n+\la+1)-\a_n(\a_n-\tfrac12t),\end{split}\]
and then substituting these into (\ref{eq:toda}) gives
\begin{subequations}\label{eq:abdisys6ab}\begin{align}\label{eq:abdisys6a}
\deriv{\a_n}t&=-\a_n(\a_n-\tfrac12t)-2\b_n+\tfrac12(2n+\la+1),\\
\label{eq:abdisys6b}
\deriv{\b_n}t&=(\a_n-\tfrac12t)\b_n - \frac{(2\b_n-n)(2\b_n-n-\la)}{2(2\a_n-t)}.
\end{align}\end{subequations}
Solving (\ref{eq:abdisys6a}) for $\b_n$ yields
\begin{equation}\label{eq:bn}
\b_n=\tfrac12\deriv{\a_n}t+\tfrac12\a_n(\a_n-\tfrac12t)-\tfrac14(2n+\la+1),
\end{equation}
and then substituting this into (\ref{eq:abdisys6b}) yields
\[\deriv[2]{\a_n}t=\frac{1}{2\a_n-t}\left(\deriv{\a_n}t-\tfrac12\right)^2 +\tfrac32\a_n^3-\tfrac54t\a_n^2
+\tfrac14(t^2-4n-2-2\la)\a_n  +\tfrac14t(2n+\la+1)-\frac{\la^2}{4(2\a_n-t)},
\]
Making the transformation (\ref{eq:scLag3a}) in this equation
yields equation (\ref{eq:scLag3}), 
which is \PIV\ (\ref{eq:PIV}) with parameters given by (\ref{eq:scLag4}). 
Finally making the transformation (\ref{eq:scLag3a}) in (\ref{eq:bn}) yields (\ref{eq:scLag3b}), as required.
}\end{proof}

\begin{remarks}{\rm\ 
\begin{enumerate}
\item Filipuk, van Assche and Zhang \cite{refFvAZ}, who considered orthonormal polynomials rather than monic orthogonal polynomials, proved the result (\ref{eq:scLag3a}) for $\a_n(t)$. However Filipuk, van Assche and Zhang \cite{refFvAZ} did not give an explicit expression for $\a_n(t)$, which we do below. 
\item From Theorem \ref{thm:p4sf} we see that the parameters (\ref{eq:scLag4}) satisfy (\ref{eq:PT.SF.eq40a}) with $\ep=-1$, and therefore satisfy the condition given in Theorem \ref{thm:p4sf} for \PIV\ to have solutions expressible in terms of {parabolic cylinder functions}.

\item If $q_n$ is a solution of equation (\ref{eq:scLag3}) then the solutions $q_{n+1}$ and $q_{n-1}$ are given by
\[\begin{split}
q_{n+1}&=\frac{\ds\left(q_n'-q_n^2-2zq_n\right)^2-4\la^2}{\ds 2q_n\left(q_n'-q_n^2-2zq_n+4n+2\la+4 \right)},\qquad
q_{n-1}=-\frac{\ds\left(q_n'+q_n^2+2zq_n\right)^2-4\la^2}{\ds 2q_n\left(q_n'+q_n^2+2zq_n-4n -2\la\right)},
\end{split}\] where $'\equiv\d/\d z$,
which are special cases of the Schlesinger transformations $\mathcal{R}_5$ (\ref{STR5}) and $\mathcal{R}_7$ (\ref{STR7}), respectively.

\item From Theorem \ref{PIVtau}, we see that the {parabolic cylinder function} solutions of equation (\ref{eq:scLag3}) are given by
\begin{equation*}\label{eq:PT.SF.eq41a0}q_n(z)=-2z+\deriv{}{z}\ln\frac{\tau_{n+1,\la}(z)}{\tau_{n,\la}(z)},\end{equation*}
where
\begin{equation*}\label{eq:PT.SF.eq41a1}
\tau_{n,\la}(z)=\W\left(\psi_{\la},\deriv{\psi_{\la}}z,\ldots,\deriv[n-1]{\psi_{\la}}z\right),\qquad \tau_{0,\la}(z)=1,
\end{equation*}
and $\psi_{\la}(z)$ satisfies
\begin{equation*}\label{eq:PT.SF.eq41a2}\deriv[2]{\psi_{\la}}z-2z\deriv{\psi_{\la}}z-2(\la+1)\psi_{\la}=0,\end{equation*}
which is equation (\ref{eq:PT.SF.eq41a}) with $\nu=-\la-1$ and $\ep=1$.
This equation has general solution
\begin{equation*}\label{eq:PT.SF.eq41a3}\psi_{\la}(z)=\begin{cases} \left\{C_1D_{-\la-1}\big(\sqrt2\,z\big)+C_2D_{-\la-1}\big(-\sqrt2\,z\big)\right\}\exp\big(\tfrac12z^2\big),\quad
&\mbox{\rm if}\quad \la\not\in\N,\\[2.5pt]
\ds C_1(-\i)^{m}H_{m}(\i z)\exp(z^2)+ C_2\deriv[m]{}{z}\left\{\erfc(z)\exp(z^2)\right\}, &\mbox{\rm if}\quad \la=m\in\N,
\end{cases}\end{equation*}
with $C_1$ and $C_2$ arbitrary constants, where $D_{\nu}(\zeta)$ is the {parabolic cylinder function}, $H_m(\zeta)$ the Hermite polynomial (\ref{def:hermite}), and $\erfc(z)$ the complementary error function (\ref{def:erfc}). 
\end{enumerate}}\end{remarks}

The system (\ref{eq:abdisys6ab}) satisfied by the recurrence coefficients $\a_n(t)$ and $\b_n(t)$ is equivalent to the Hamiltonian system (\ref{eq:PT.HM.DE423}) associated with \PIV, as shown in the following Theorem.

\begin{theorem}{The system (\ref{eq:abdisys6ab}) is equivalent to the Hamiltonian system (\ref{eq:PT.HM.DE423}) associated with \PIV. }\end{theorem}

\begin{proof}{If in the system (\ref{eq:abdisys6ab}) we make the transformation
$$\a_n(t)=\tfrac12q_n(z)+\tfrac12t,\qquad \b_n(t)= -\tfrac12q_n(z)p_n(z)+\tfrac12(n+\la), \qquad z=\tfrac12t,$$
then $q_n(z)$ and $p_n(z)$ satisfy the system 
\begin{subequations}\label{eq:abdisys6pq}\begin{align}
\deriv{q_n}{z}&=4q_np_n-q_n^2-2zq_n-2\la,\\
\deriv{p_n}{z}&=-2p_n^2+2p_nq_n+2zp_n-n-\la,
\end{align}\end{subequations}
which is the system (\ref{eq:PT.HM.DE423}) with $\th_0=\la$ and $\th_{\infty}=\la+n$. Conversely making the transformation 
\[q_n(z)=2\a_n(t)-t,\qquad p_n(z)=-\frac{2\b_n(t)-n-\la}{2\a_n(t)-t},\qquad t=2z,\]
in the system (\ref{eq:abdisys6pq}) yields the system (\ref{eq:abdisys6ab}).}\end{proof}

Our main objective is to obtain explicit expressions for the coefficients $\a_n(t)$ and $\b_n(t)$ in the recurrence relation (\ref{eq:scLag2}). First we derive an explicit expression for the moment $\mu_0(t;\la)$. 

\begin{theorem}{For the semi-classical Laguerre weight (\ref{eq:scLag1}), the moment $\mu_0(t;\la)$ is given by
\begin{equation}\label{def:mu0}
\mu_0(t;\la)=\begin{cases} 
\ds\frac{\Gamma(\la+1)\exp\left(\tfrac18t^2\right)}{2^{(\la+1)/2}}\,\WhitD{-\la-1}\big(-\tfrac12\sqrt2\,t\big),&\mbox{if}\quad \la\not\in\N,\\[5pt]
\ds\tfrac12\sqrt{\pi}\, \deriv[m]{}{t}\left\{\exp\big(\tfrac{1}{4}t^{2}\big)\left[1+\erf(\tfrac12t)\right]\right\},
&\mbox{if}\quad \la=m\in\N,\end{cases}\end{equation}
with $D_\nu(\zeta) $ the {parabolic cylinder function} and $\erf(z)$ the error function. Further $\mu_0(t;\la)$ satisfies the equation
\begin{equation}
\deriv[2]{\mu_0}{t}-\tfrac12t\deriv{\mu_0}{t}-\tfrac12(\la+1)\mu_0=0.\label{eq:mu0}
\end{equation}
}\end{theorem}

\begin{proof}{The {parabolic cylinder function} $D_\nu(\zeta) $, with $\nu\not\in\Z$, has the integral representation\ \cite[\S12.5(i)]{refDLMF}
$$D_\nu(\zeta)=\frac{\exp(-\tfrac14\zeta^2)}{\Gamma(-\nu)}\int_0^\infty s^{-\nu-1}\exp(-\tfrac12s^2-\zeta s)\,\d s,\qquad \Re(\nu)<0.$$
For the semi-classical Laguerre weight (\ref{eq:scLag1}), the moment $\mu_0(t;\la)$, with $\la\not\in\N$, 
is given by
\begin{align*}
\mu_0(t;\la)&=\int_0^\infty x^\la \exp(-x^2+xt)\,\d x\\
&=2^{-(\la+1)/2}\int_0^\infty s^{\la} \exp\left(-\tfrac12s^2+\tfrac12\sqrt2\,t\,s\right)\,\d s\\
&= \frac{\Gamma(\la+1)\exp\left(\tfrac18t^2\right)}{2^{(\la+1)/2}}\,\WhitD{-\la-1}\big(-\tfrac12\sqrt2\,t\big)
\end{align*}
as required. If $m\in\N$, then the {parabolic cylinder function} $D_{-m-1}(\zeta)$ is given by
\[\mathop{D_{-m-1}}(\zeta)=\sqrt{\frac{\pi}{2}}\frac{(-1)^{m}}{m!}\exp(-\tfrac{1}{4}\zeta^{2}) \deriv[m]{}{\zeta}\left\{\exp(\tfrac{1}{2}\zeta^{2})\erfc\left(\tfrac12\sqrt{2}\,\zeta\right)\right\},\] with $\erfc(z)$ the complementary error function\ \cite[\S12.7(ii)]{refDLMF}.
Since $\erfc(-z)=1+\erf(z)$, then $\mu_0(t;m)$, with $m\in\N$, is given by
\[\mu_0(t;m)=\tfrac12\sqrt{\pi}\, \deriv[m]{}{t}\left\{\exp\big(\tfrac{1}{4}t^{2}\big)\left[1+\erf(\tfrac12t)\right]\right\},\] 
as required.
Further, the {parabolic cylinder function} $D_\nu(\zeta)$ satisfies equation (\ref{eq:pcf}) 
and so from (\ref{def:mu0}) it follows that the moment $\mu_0(t;\la)$ satisfies equation (\ref{eq:mu0}), as required.}\end{proof}

\begin{corollary}{If $\mu_0(t;\la)$ is given by (\ref{def:mu0}) and $\ph_{\nu}(z;\ep)$ by (\ref{eq:PT.SF.eq44nu}), then
\begin{equation*}\label{lag:mu0}
\mu_0(t;\la)=\ph_{-\la-1}(\tfrac12t;1),
\end{equation*}
with $C_1=0$ and $C_2=\Gamma(\la+1)/2^{(\la+1)/2}$.}\end{corollary}
\begin{proof}{The result is easily shown by comparing (\ref{def:mu0}) and (\ref{eq:PT.SF.eq44nu}).}\end{proof}

Having obtained an explicit expression for $\mu_0$ we can now derive explicit expressions for the Hankel determinant $\Delta_n(t)$ and the coefficients $\a_n(t)$ and $\b_n(t)$ in the recurrence relation (\ref{eq:scLag2}).

\begin{theorem}{The Hankel determinant $\Delta_n(t)$ is given by 
\begin{equation}\label{scLagHank}\Delta_n(t)=
\W\left(\mu_0,\deriv{\mu_0}t,\ldots,\deriv[n-1]{\mu_0}t\right),\end{equation}
with $\mu_0$ given by (\ref{def:mu0}).}\end{theorem}

\begin{proof}{This is an immediate consequence of  Theorem \ref{thm:2.1}.}\end{proof}

\begin{theorem}{The coefficients $\a_n(t)$ and $\b_n(t)$ in the recurrence relation (\ref{eq:scLag2})
associated with monic polynomials orthogonal with respect to the semi-classical Laguerre weight (\ref{eq:scLag1}) are given by
\begin{equation*}\label{scLagab} \a_n(t) =\deriv{}{t}\ln \frac{\Delta_{n+1}(t)}{\Delta_{n}(t)},\qquad 
\b_n(t) =\deriv[2]{}{t}\ln\Delta_n(t) ,
\end{equation*}
where $\Delta_n(t)$ is the Hankel determinant given by (\ref{scLagHank}), with $\mu_0$ given by (\ref{def:mu0}).
}\end{theorem}
\begin{proof}{This is an immediate consequence of Theorems \ref{thm:2.1} and \ref{thm:anbn}.}\end{proof}

Furthermore we can relate the Hankel determinant $\Delta_n(t)$ given by (\ref{scLagHank}) to the $\tau$-function $\tau_{n,\nu}(z;\ep)$ given by (\ref{eq:taudet1}).

\begin{theorem}{If $\Delta_n(t)$ is given by (\ref{scLagHank}) and $\tau_{n,\nu}(z;\ep)$ by (\ref{eq:taudet1}), with $$\ds\ph_{-\la-1}(z)=\frac{\Gamma(\la+1)\exp\left(\tfrac12z^2\right)}{2^{(\la+1)/2}}
\,\WhitD{-\la-1}\left(-\sqrt2\,z\right),$$ then
\begin{equation}\label{eq:4.20}
\Delta_n(t)=\left.\frac{ \tau_{n,-\la-1}(z;1)}{2^{n(n-1)}}\right|_{z=t/2}.
\end{equation} }\end{theorem}
\begin{proof}{The result is easily shown by comparing (\ref{scLagHank}) and (\ref{eq:taudet1}).}\end{proof}

\begin{theorem}{The function $\ds S_n(t)=\deriv{}{t}\ln\Delta_n(t)$, with $\Delta_n(t)$ given by (\ref{scLagHank}), satisfies the second-order, second-degree equation
\begin{equation}\label{eq:scLag5}
4\left(\deriv[2]{S_n}{t}\right)^2-\left(t\deriv{S_n}t-S_n\right)^2+4\deriv{S_n}t\left(2\deriv{S_n}t-n\right)\left(2\deriv{S_n}t-n-\la\right)=0.
\end{equation}
}\end{theorem}
\begin{proof}{Setting $\nu=-\la-1$ and $\ep=1$ in (\ref{eq:PT.HM.DE44c}) gives
$$\sigma(z;\la,n+\la) = \deriv{}{z}\ln \tau_{n,-\la-1}(z;1)-2(n+\la)z,$$ 
and so if $\ds S_n(t)=\deriv{}{t}\ln\Delta_n(t)$ then from (\ref{eq:4.20}) we see that 
\begin{equation*}\label{eq:4.21} \sigma(z;\la,n+\la)=2S_n(t)-(n+\la)t,\qquad z=\tfrac12t.\end{equation*}
Making this transformation in \sPIV\ (\ref{eq:PT.HM.DE44}) with $\vth_0=\la$ and $\vth_\infty=n+\la$, i.e.
\[
\left(\deriv[2]{\sigma}{z}\right)^2-4\left(z\deriv{\sigma}z-\sigma\right)^2
+4\deriv{\sigma}z\left(\deriv{\sigma}z+2\la\right)\left(\deriv{\sigma}z+2n+2\la\right)=0,
\] 
yields  (\ref{eq:scLag5}), as required.
}\end{proof}

\begin{remark}{Differentiating (\ref{eq:scLag5}) and letting $\ds S_n(t)=\deriv{}{t}\ln\Delta_n(t)$ yields the fourth-order, bi-linear equation
\begin{align*}\Delta_n\deriv[4]{\Delta_n}t&-4\deriv[3]{\Delta_n}t\deriv{\Delta_n}t+3\left(\deriv[2]{\Delta_n}t\right)^2
-(\tfrac14t^2+4n+2\la)\left[\Delta_n\deriv[2]{\Delta_n}t-\left(\deriv{\Delta_n}t\right)^2\right]\nonumber\\
&+\tfrac14t\Delta_n\deriv{\Delta_n}t+\tfrac12n(n+\la)\Delta_n^2=0,
\end{align*}as is easily verified.}\end{remark}

\begin{theorem}{Suppose $\Psi_{n,\la}(z)$ is given by
\[\label{eq:PT.SF.eq411}
\Psi_{n,\la}(z)=\W\left(\psi_{\la},\deriv{\psi_{\la}}z,\ldots,\deriv[n-1]{\psi_{\la}}z\right),\qquad \Psi_{0,\la}(z)=1,
\]
where 
\[\label{eq:PT.SF.eq413}
\psi_{\la}(z)=\begin{cases} D_{-\la-1}\big(-\sqrt2\,z\big)\exp\big(\tfrac12z^2\big),\quad
&\mbox{\rm if}\quad \la\not\in\N,\\[2.5pt]
\ds \deriv[m]{}{z}\left\{\big[1+\erf(z)\big]\exp(z^2)\right\}, &\mbox{\rm if}\quad \la=m\in\N,
\end{cases}\]
with $D_{\nu}(\zeta)$ is {parabolic cylinder function} and $\erfc(z)$ the complementary error function (\ref{def:erfc}). 
Then coefficients $\a_n(t)$ and $\b_n(t)$ in the recurrence relation (\ref{eq:scLag2}) 
associated with the semi-classical Laguerre weight (\ref{eq:scLag1}) are given by
\[\begin{split}\a_n(t)&=\tfrac12q_n(z)+\tfrac12t,\\ 
\b_n(t)&=-\tfrac18\deriv{q_n}z-\tfrac18 q_n^2(z) -\tfrac14zq_n(z)+\tfrac14\la+\tfrac12n,
\end{split}\] 
with $z=\tfrac12t$, where 
\[\label{eq:PT.SF.eq410}q_n(z)=-2z+\deriv{}{z}\ln\frac{\Psi_{n+1,\la}(z)}{\Psi_{n,\la}(z)},\]
which satisfies \PIV\ (\ref{eq:PIV}), with parameters $(A,B)=(2n+\la+1,-2\la^2)$.
}\end{theorem}

In Appendix 1 we give the first few recurrence coefficients for the semi-classical Laguerre weight (\ref{eq:scLag1}) and the first few monic polynomials generated using the recurrence relation (\ref{eq:scLag2}).

\section{\label{sec:asymp}Asymptotic expansions}
In this section we derive asymptotic expansions for the the moment $\mu_0(t;\la)$, see Lemma \ref{lem:5.1} below, the Hankel determinant $\Delta_n(t)$, see Lemma \ref{lem:5.2} below, and the recurrence coefficients $\a_n(t)$ and $\b_n(t)$, see Lemma \ref{lem:5.3} below.

\begin{lemma}{\label{lem:5.1}As $t\to\infty$, the moment $\mu_0(t;\la)$ has the asymptotic expansion
\begin{equation}\label{mu0:asymp}
\mu_0(t;\la) \sim\sqrt{\pi}\,(\tfrac12t)^{\la}\exp\left(\tfrac14t^2\right)\sum_{n=0}^\infty \frac{\Gamma(\la+1)}{\Gamma(\la-n+1)\,n!\,t^{2n}}.
\end{equation}}\end{lemma}

\begin{proof}{Since the {parabolic cylinder function} $D_\nu(\zeta)$ has the asymptotic expansion 
$$D_{\nu}(\zeta)\sim\frac{\sqrt{2\pi}\,(-1)^{\nu+1}}{\Gamma(-\nu)\zeta^{\nu+1}}\exp(\tfrac14\zeta^2)
\sum_{n=0}^\infty \frac{\left(\nu+1\right)_{2n}}{n!\,(2\zeta^2)^n},\qquad\mbox{as}\quad \zeta\to-\infty,$$ 
with $(\b)_n=\Gamma(\b+n)/\Gamma(\b)$ the Pochhammer symbol, then
\[\begin{split} \mu_0(t)&= \frac{\Gamma(\la+1)\exp\left(\tfrac18t^2\right)}{2^{(\la+1)/2}}\,\WhitD{-\la-1}\big(-\tfrac12\sqrt2\,t\big)\\
&\sim \frac{\Gamma(\la+1)\exp\left(\tfrac18t^2\right)}{2^{(\la+1)/2}}\,\frac{\sqrt{2\pi}\,t^{\la}\exp\left(\tfrac18t^2\right)}{\Gamma(\la+1)2^{\la/2}}
\sum_{n=0}^\infty \frac{(-\la)_{2n}}{n!\,t^{2n}}\\
&=\sqrt{\pi}\,(\tfrac12t)^{\la}\exp\left(\tfrac14t^2\right)\sum_{n=0}^\infty \frac{\Gamma(\la+1)}{\Gamma(\la-n+1)\,n!\,t^{2n}},\end{split}\]
as required, since \[(-\la)_{2n}=\frac{\Gamma(2n-\la)}{\Gamma(-\la)}=\la(\la-1)\ldots(\la-2n+1)=\frac{\Gamma(\la+1)}{\Gamma(\la-n+1)}.\]
}\end{proof}

\begin{lemma}{\label{lem:5.2}As $t\to\infty$, the Hankel determinant $\Delta_n(t)$ has the asymptotic expansion
\begin{equation} \label{def:deltan}
\Delta_n(t)= c_n\pi^{n/2}(\tfrac12t)^{n\la}\exp\big(\tfrac14nt^2\big)\left\{1+\frac{n\la(\la-n)}{t^2}+\mathcal{O}\left({t^{-4}}\right)\right\},
\end{equation}
with $c_n$ a constant, and $S_n(t)$ has the asymptotic expansion
\begin{equation} \label{def:Sn}
S_n(t)= \frac{nt}{2}+\frac{n\la}{t}+\frac{2n\la(n-\la)}{t^3} 
+\mathcal{O}\left({t^{-5}}\right).
\end{equation}}\end{lemma}

\begin{proof}{To prove (\ref{def:deltan}) we shall use Mathematical induction. Since $\Delta_n$ satisfies the Toda equation (\ref{eq:toda1}) then
\begin{equation}\label{eq:deltan} \Delta_{n+1}=\frac{1}{\Delta_{n-1}}\left\{\Delta_n\deriv[2]{\Delta_n}{t}-\left(\deriv{\Delta_n}{t}\right)^2\right\}.\end{equation}
By definition $\Delta_0=1$ and from (\ref{mu0:asymp})
\begin{equation}\label{eq:delta1}
\Delta_1=\mu_0=\sqrt{\pi}\,(\tfrac12t)^{\la}\exp\left(\tfrac14t^2\right) \left\{1+\frac{\la(\la-1)}{t^2}+\mathcal{O}\left({t^{-4}}\right)\right\}.
\end{equation} as $t\to\infty$, and so (\ref{eq:deltan}) with $n=1$ gives
\[\begin{split}\Delta_2=&\left\{\Delta_1\deriv[2]{\Delta_1}{t}-\left(\deriv{\Delta_1}{t}\right)^2\right\}
=\tfrac12\pi(\tfrac12t)^{2\la}\exp\big(\tfrac12t^2\big)\left\{1+\frac{2\la(\la-2)}{t^2}+\mathcal{O}\left({t^{-4}}\right)\right\},
\end{split}\] as $t\to\infty$.
Assuming (\ref{def:deltan}) then
\[\left\{\Delta_n\deriv[2]{\Delta_n}{t}-\left(\deriv{\Delta_n}{t}\right)^2\right\} = \tfrac12n
c_n^2\pi^{n}(\tfrac12t)^{2n\la}\exp\big(\tfrac12nt^2\big)\left\{1+\frac {2\la(n\la-n^2-1) }{{t}^2}+\mathcal{O}\left({t^{-4}}\right)\right\},\]
as $t\to\infty$, and so \[\begin{split}
 \Delta_{n+1}&=\frac{1}{\Delta_{n-1}}\left\{\Delta_n\deriv[2]{\Delta_n}{t}-\left(\deriv{\Delta_n}{t}\right)^2\right\}\\
 &=\frac{nc_n^2}{2c_{n-1}}\pi^{(n+1)/2}(\tfrac12t)^{(n+1)\la}\exp\big\{\tfrac14(n+1)t^2\big\}\\
 &\qquad\qquad\times
 \left\{1+\frac {2\la(n\la-n^2-1) }{{t}^2}+\mathcal{O}\left({t^{-4}}\right)\right\}\left\{1-\frac{(n-1)\la(\la-n+1)}{t^2}+\mathcal{O}\left({t^{-4}}\right)\right\}\\
  &=c_{n+1}\pi^{(n+1)/2}(\tfrac12t)^{(n+1)\la}\exp\big\{\tfrac14(n+1)t^2\big\}
 \left\{1+\frac{(n+1)\la(\la-n-1)}{t^2}+\mathcal{O}\left({t^{-4}}\right)\right\}
\end{split}\] as $t\to\infty$, where $c_{n+1}= \tfrac12nc_n^2/c_{n-1}$, as required. Solving the recurrence relation
\[c_{n+1}c_{n-1}=\tfrac12nc_n^2,\qquad c_0=1,\quad c_1=1,\] gives
\[c_n= \frac{1}{2^{n(n-1)/2}}\prod _{k=0}^{n-1}{k!}.\]
Since $\ds S_n=\deriv{}{t}\ln\Delta_n$ then the asymptotic expansion (\ref{def:Sn}) is easily derived from (\ref{def:deltan}).
}\end{proof}

\begin{lemma}{\label{lem:5.3}As $t\to\infty$, the recurrence coefficients $\a_n(t)$ and $\b_n(t)$ have the asymptotic expansions
\begin{align*}
\a_n(t)&=\frac{t}{2}+\frac{\la}{t}+\frac{2\la(2n-\la+1)}{t^3}+\mathcal{O}\left({t^{-5}}\right),\\ 
\b_n(t)&=\frac{n}{2}-\frac{n\la}{t^2}-\frac{6n\la(n-\la)}{t^4}+\mathcal{O}\left({t^{-6}}\right).
\end{align*}}\end{lemma}

\begin{proof}{By definition\[
\a_n(t)=\deriv{}{t}\ln \frac{\Delta_{n+1}(t)}{\Delta_{n}(t)} =S_{n+1}(t)-S_n(t),\qquad \b_n(t)=\deriv[2]{}{t}\ln\Delta_n(t)=\deriv{S_n}{t},
\]
and so
\[\begin{split}\a_n(t)&=\frac{t}{2}+\frac{\la}{t}+\frac{2\la(2n-\la+1)}{t^3}+\mathcal{O}\left(t^{-5}\right),\qquad
\b_n(t)=\frac{n}{2}-\frac{n\la}{t^2}-\frac{6n\la(n-\la)}{t^4}+\mathcal{O}\left(t^{-6}\right),
\end{split}\]
as $t\to\infty$, as required. Consequently 
\[\lim_{t=\infty}\a_n(t)=\tfrac12{t},\qquad \lim_{t=\infty}\b_n(t) = \tfrac12{n}.\]
}\end{proof}

 \section{\label{sec:scherm}Semi-classical Hermite weight}
In this section we are concerned with the semi-classical Hermite weight
\begin{equation}\label{eq:scHerm1}
\w(x;t)=|x|^{\la}\exp(-x^2+tx),\qquad x,t\in\R,\quad \la>-1,
\end{equation}
which is an extension of the semi-classical Laguerre weight (\ref{eq:scLag1}) to the whole real line, where we have ensured that the weight is positive by using $|x|^{\la}$ rather than $x^{\la}$.
Monic orthogonal polynomials associated with the semi-classical Hermite weight (\ref{eq:scHerm1}) satisfy the recurrence relation 
\begin{equation}\label{eq:scHerm2}
xP_n(x;t)=P_{n+1}(x;t)+\a_n(t)P_n(x;t)+\b_{n}(t)P_{n-1}(x;t),\end{equation}
and our interest is in obtaining explicit expressions for the coefficients $\a_n(t)$ and $\b_n(t)$ in (\ref{eq:scHerm2}).

First we evaluate the moment $\mu_0(t;\la)$. 

\begin{theorem}{For the semi-classical Hermite weight (\ref{eq:scHerm1}), the moment $\mu_0(t;\la)$ is given by
\begin{equation} \label{herm:mu0}
\mu_0(t;\la) = \begin{cases}
\ds\frac{\Gamma(\la+1)\exp(\tfrac18t^2)}{2^{(\la+1)/2}}\Big\{\WhitD{-\la-1}\big(-\tfrac12\sqrt2\,t\big) 
+\WhitD{-\la-1}\big(\tfrac12\sqrt2\,t\big)\Big\},&\mbox{if}\quad \la\not\in\N,\\[5pt]
\ds\sqrt{\pi}\,\big(-\tfrac12\i\big)^{2m}H_{2m}\big(\tfrac12\i t\big)\exp\big(\tfrac14t^2\big), & \mbox{if}\quad\la=2m,\\[5pt]
\ds\sqrt{\pi}\;\deriv[2m+1]{}{t}\left\{\erf(\tfrac12t)\,\exp\big(\tfrac{1}{4}t^{2}\big)\right\}, & \mbox{if}\quad\la=2m+1,
\end{cases}\end{equation} with $m\in\N$,
where $\WhitD{\nu}(z)$ is the parabolic cylinder function, $H_n(z)$ is the Hermite polynomial and $\erf(z)$ is the error function.}\end{theorem}

\begin{proof}{If $\la\not\in\N$, then the moment $\mu_0(t;\la)$ is given by
\begin{align*}
\mu_0(t;\la)&=\int_{-\infty}^\infty \w(x;t)\,\d x =\int_{-\infty}^\infty  |x|^{\la}\exp(-x^2+tx)\,\d x \nonumber\\
&= \int_{0}^\infty x^{\la}\exp(-x^2+tx)\,\d x+
\int_{0}^\infty x^{\la}\exp(-x^2-tx)\,\d x\nonumber\\
&=\frac{\Gamma(\la+1)\exp(\tfrac18t^2)}{2^{(\la+1)/2}}\Big\{\WhitD{-\la-1}\big(-\tfrac12\sqrt2\,t\big) 
+\WhitD{-\la-1}\big(\tfrac12\sqrt2\,t\big)\Big\},
\end{align*} as required.
If $\la=2m$, with $m\in\N$, then 
\[
\mu_0(t;2m)= \int_{-\infty}^\infty x^{2m}\exp(-x^2+tx)\,\d x=
\sqrt{\pi}\,\big(-\tfrac12\i\big)^{2m}H_{2m}(\tfrac12\i t)\exp\big(\tfrac14t^2\big),\] 
as required, since the {Hermite polynomial}, $H_n(z)$, has the integral representation
\[H_m(z)=\frac{2^m}{\sqrt{\pi}}\int_{-\infty}^{\infty}(z+\i x)^m\exp(-x^2)\,\d x.\]
Finally if $\la=2m+1$, with $m\in\N$, then 
\[\begin{split}
\mu_0(t;2m+1)&= \int_{-\infty}^\infty x^{2m}|x|\exp(-x^2+tx)\,\d x\\
&= \deriv[2m]{}{t} \left(\int_{0}^\infty x\exp(-x^2+tx)\,\d x + \int_{0}^\infty x\exp(-x^2-tx)\,\d x\right)\\
&= \deriv[2m+1]{}{t} \left(\int_{0}^\infty \exp(-x^2+tx)\,\d x - \int_{0}^\infty \exp(-x^2-tx)\,\d x\right)\\
&= \deriv[2m+1]{}{t} \bigg(\tfrac12\sqrt{\pi}\,\left\{1+\erf(\tfrac12t)\right\}\exp\big(\tfrac{1}{4}t^{2}\big) - 
\tfrac12\sqrt{\pi}\, \left\{1-\erf(\tfrac12t)\right\}\exp\big(\tfrac{1}{4}t^{2}\big) \bigg)\\
&= \sqrt{\pi}\deriv[2m+1]{}{t} \left\{\erf(\tfrac12t)\,\exp\big(\tfrac{1}{4}t^{2}\big)\right\},
\end{split}\]
as required, since
\[\int_{0}^\infty \exp(-x^2+tx)\,\d x =\tfrac12\sqrt{\pi}\,\left\{1+\erf(\tfrac12t)\right\}\exp\big(\tfrac{1}{4}t^{2}\big).\]}\end{proof}

Next we obtain an explicit expression for the Hankel determinant $\Delta_n(t)$.

\begin{theorem}{The Hankel determinant $\Delta_n(t)$ is given by 
\begin{equation}\label{scHermHank}\Delta_n(t)=
\W\left(\mu_0,\deriv{\mu_0}t,\ldots,\deriv[n-1]{\mu_0}t\right),\end{equation}
where $\mu_0(t;\la)$ is given by (\ref{herm:mu0}).}\end{theorem}

\begin{proof}{By definition the moment $\mu_k(t;\la)$ is given by  
\[\begin{split}\mu_k(t;\la)&=\int_{-\infty}^\infty x^{k}|x|^{\la}\exp(-x^2+tx)\,\d x \\
&= \deriv[k]{}{t}\left(\int_{-\infty}^\infty |x|^{\la}\exp(-x^2+tx)\,\d x\right)=\deriv[k]{\mu_0}{t},
\end{split}\] 
and so we obtain 
\[\Delta_n(t)=\det\Big[\mu_{j+k}(t)\Big]_{j,k=0}^{n-1}
\equiv\W\left(\mu_0,\deriv{\mu_0}t,\ldots,\deriv[n-1]{\mu_0}t\right),\]
as required.}\end{proof}

Finally we obtain explicit expressions for the coefficients $\a_n(t)$ and $\b_n(t)$.

\begin{theorem}{The coefficients $\a_n(t)$ and $\b_n(t)$ in the recurrence relation (\ref{eq:scHerm2})
associated with monic polynomials orthogonal with respect to the semi-classical Hermite weight (\ref{eq:scHerm1}) are given by
\begin{equation*}\label{scHermab} \a_n(t) =\deriv{}{t}\ln \frac{\Delta_{n+1}(t)}{\Delta_{n}(t)},\qquad 
\b_n(t) =\deriv[2]{}{t}\ln\Delta_n(t) ,
\end{equation*}
where $\Delta_n(t)$ is the Hankel determinant given by (\ref{scHermHank}), with $\mu_0(t;\la)$ given by (\ref{herm:mu0}).
}\end{theorem}
\begin{proof}{This is an immediate consequence of Theorem \ref{thm:anbn}.}\end{proof}

\begin{theorem}{Suppose $\widetilde{\Psi}_{n,\la}(z)$ is given by
\[\label{eq:PT.SF.eq421}
\widetilde{\Psi}_{n,\la}(z)=\W\left(\widetilde{\psi}_{\la},\deriv{\widetilde{\psi}_{\la}}z,\ldots,\deriv[n-1]{\widetilde{\psi}_{\la}}z\right),\qquad \Psi_{0,\la}(z)=1,
\]
where 
\[\label{eq:PT.SF.eq423}
\widetilde{\psi}_{\la}(z)=\begin{cases} \left\{D_{-\la-1}\big(\sqrt2\,z\big)+D_{-\la-1}\big(-\sqrt2\,z\big)\right\}\exp\big(\tfrac12z^2\big),\quad
&\mbox{\rm if}\quad \la\not\in\N,\\[2.5pt]
\ds H_{2m}(\i z)\exp(z^2), &\mbox{\rm if}\quad \la=2m,\quad m\in\N,\\[2.5pt]
\ds \deriv[2m+1]{}{z}\left\{\erf(z)\exp(z^2)\right\}, &\mbox{\rm if}\quad \la=2m+1,\quad m\in\N,
\end{cases}\]
with $D_{\nu}(\zeta)$ is {parabolic cylinder function}, $H_m(\zeta)$ the Hermite polynomial (\ref{def:hermite}), and $\erfc(z)$ the complementary error function (\ref{def:erfc}). 
Then coefficients $\a_n(t)$ and $\b_n(t)$ in the recurrence relation (\ref{eq:scHerm2}) 
associated with the semi-classical Hermite weight (\ref{eq:scHerm1}) are given by
\[\begin{split}\a_n(t)&=\tfrac12q_n(z)+\tfrac12t,\\ \label{eq:scHerm31b}
\b_n(t)&=-\tfrac18\deriv{q_n}z-\tfrac18 q_n^2(z) -\tfrac14zq_n(z)+\tfrac14\la+\tfrac12n,
\end{split}\] 
with $z=\tfrac12t$, where 
\[\label{eq:PT.SF.eq420}q_n(z)=-2z+\deriv{}{z}\ln\frac{\widetilde{\Psi}_{n+1,\la}(z)}{\widetilde{\Psi}_{n,\la}(z)},\]
which satisfies \PIV\ (\ref{eq:PIV}), with parameters given by $(A,B)=(2n+\la+1,-2\la^2)$.
}\end{theorem}

In Appendix 2 we give the first few recurrence coefficients for the semi-classical Hermite weight (\ref{eq:scHerm1}), in the case when $\la=2$ (so the recurrence coefficients are rational functions of $t$), and the first few monic polynomials generated using the recurrence relation (\ref{eq:scHerm2}).

\section{Discussion}\label{sec:dis}
In this paper we have studied semi-classical Laguerre polynomials which are orthogonal polynomials that satisfy three-term recurrence relations whose coefficients depend on a parameter. We have shown that the coefficients in these recurrence relations can be expressed in terms of Wronskians of parabolic cylinder functions. These Wronskians also arise in the description of special function solutions of the fourth \p\ equation and the second-order, second-degree equation satisfied by the associated Hamiltonian function. Further we have shown similar results hold for semi-classical Hermite polynomials. The link between the semi-classical orthogonal polynomials and the special function solutions of the \p\ equations is the moment for the associated weight which enables the Hankel determinant to be written as a Wronskian. In our opinion, this illustrates the increasing significance of the \p\ equations in the field of orthogonal polynomials and special functions.

\section*{Acknowledgements} 
We thank the London Mathematical Society for the support through a ``Research in Pairs" grant. 
PAC thanks Ana Loureiro, Paul Nevai, James Smith and Walter van Assche
for their helpful comments and illuminating discussions. We also thank the referees for helpful suggestions and additional references.

\subsection*{Appendix 1. Recurrence coefficients and polynomials for the semi-classical Laguerre weight}
For the semi-classical Laguerre weight the first few recurrence coefficients are given by
\[\begin{split}
\alpha_0(t)&=\tfrac12t- \frac{D_{-\la}\big(-\tfrac12\sqrt2\,t\big)}{D_{-\la-1}\big(-\tfrac12\sqrt2\,t\big)}\equiv\Psi_\nu(t),\\
\alpha_1(t)&= \tfrac12t-\Psi_\nu(t)-\frac{\Psi_\nu(t)}{2\Psi_\nu^2(t)-t\Psi_\nu(t)-\la-1},\\
\alpha_2(t)&=\tfrac12t+\frac{2\la+4}{t}+\frac{\Psi_\nu(t)}{2\Psi_\nu^2(t)-t\Psi_\nu(t)-\la-1},\\
&\qquad -\frac{2[(\la+1)t^2+4(\la+2)(2\la+3)]\Psi_\nu^2(t)-(\la+1)t[t^2+2(4\la+9)]\Psi_\nu(t)- (\la+1)^2[t^2+8(\la+2)]}{2t\big[ 2t\Psi_\nu^{3}(t)-({t}^{2}-4\la-6) \Psi_\nu^{2}(t) -3(\la+1)t \Psi_\nu(t) -2(\la+1) ^{2} \big] },\\[5pt]
\beta_1(t)&=-\Psi_\nu^2(t)+\tfrac12t\Psi_\nu(t)+\tfrac12(\la+1),\\
\beta_2(t)&=-\frac{2t\Psi_\nu^3(t)-(t^2-4\la-6)\Psi_\nu^2(t)-3(\la+1)t\Psi_\nu(t)-2(\la+1)^2}{2\big[\Psi_\nu^2(t)-\tfrac12t\Psi_\nu(t)-\tfrac12(\la+1)\big]^2},
\end{split}\]
and the first few monic orthogonal polynomials are given by
\[\begin{split}
P_1(x;t)&={x}-\Psi_\nu\\
P_2(x;t)&={x}^{2}-\frac{2t\Psi_\nu^2-(t^2+2)\Psi_\nu-(\la+1)t}{2\big[\Psi_\nu^2-\tfrac12t\Psi_\nu-\tfrac12(\la+1)\big]}\,x
-\frac{2(\la+2)\Psi_\nu^2-(\lambda+1)\Psi_\nu-(\lambda+1)^2}{2\big[\Psi_\nu^2-\tfrac12t\Psi_\nu-\tfrac12(\la+1)\big]}\\
P_3(x;t)&=x^3-\left\{\frac{4(t^2+2\la+4)\Psi_\nu^3-2t(t^2-\la-1)\Psi_\nu^2-(\la+1)(5t^2+4\lambda+6)\Psi_\nu-3(\la+1)^2t}
{2\big[ 2t\Psi_\nu^{3}-({t}^{2}-4\la-6) \Psi_\nu^{2} -3(\la+1)t \Psi_\nu -2(\la+1) ^{2} \big]}\right\}x^2\\
&\qquad +\left\{\frac{2t(t^2+2\la+4)\Psi_\nu^3-\big[t^4+4(2\la+5)(\la+2)\big]\Psi_\nu^2-2(\la+1)t(t^2-\la-5)\Psi_\nu-(\la+1)^2(t^2-4\la-12)}{4\big[ 2t\Psi_\nu^{3}-({t}^{2}-4\la-6) \Psi_\nu^{2} -3(\la+1)t \Psi_\nu -2(\la+1) ^{2} \big]}\right\}x\\
&\qquad +\frac{2\big[(\la+1)t^2+4(\la+2)^2\big]\Psi_\nu^3-(\la+1)t(t^2+2\la+8)\Psi_\nu^2-2(\la+1)^2(t^2+2\la+5)\Psi_\nu-(\la+1)^3t}
{4\big[ 2t\Psi_\nu^{3}-({t}^{2}-4\la-6) \Psi_\nu^{2} -3(\la+1)t \Psi_\nu -2(\la+1) ^{2} \big]}
\end{split}\]

\subsection*{Appendix 2. Recurrence coefficients and polynomials for the semi-classical Hermite weight}
For the semi-classical Hermite weight $x^{2}\exp(-x^2+tx)$ the first few recurrence coefficients are given by
\[\begin{split}
\a_0(t)&=\tfrac12t+{\frac{2t}{t^{2}+2}},\\
\a_1(t)&=\tfrac12t+{\frac{4t^{3}}{t^{4}+12}}-{\frac{2t}{t^{2}+2}},\\
\a_2(t)&=\tfrac12t+{\frac{6t( t^{4}-4t^{2}+12)}{t^{6}-6t^{4}+36t^{2}+72}}-{\frac{4t^{3}}{t^{4}+12}},\\
\a_3(t)&=\tfrac12t+{\frac{8t^{3}(t^{4}-12t^{2}+60)}{t^{8}-16t^{6}+120t^{4}+720}}-{\frac{6t( t^{4}-4t^{2}+12)}{t^{6}-6t^{4}+36t^{2}+72}},\\
\a_4(t)&=\tfrac12t+{\frac{10t(t^{8}-24t^{6}+216t^{4}-480t^{2}+720)}{t^{10}-30t^{8}+360t^{6}-1200t^{4}+3600t^{2}+7200}}
-{\frac{8t^{3}(t^{4}-12t^{2}+60)}{t^{8}-16t^{6}+120t^{4}+720}},\\
\a_5(t)&=\tfrac12t+{\frac{12t^{3}(t^{8}-40t^{6}+600t^{4}-3360t^{2}+8400)}{t^{12}-48t^{10}+900t^{8}-6720t^{6}+25200t^{4}+100800}}
-{\frac{10t(t^{8}-24t^{6}+216t^{4}-480t^{2}+720)}{t^{10}-30t^{8}+360t^{6}-1200t^{4}+3600t^{2}+7200}},\\[5pt]
\b_1(t)&=\tfrac12-{\frac{2({t}^{2}-2)}{({t}^{2}+2)^{2}}},\\
\b_2(t)&=1-{\frac{4{t}^{2}({t}^{2}-6)({t}^{2}+6)}{({t}^{4}+12)^{2}}},\\
\b_3(t)&=\tfrac32-{\frac{6({t}^{4}-12{t}^{2}+12)({t}^{6}+6{t}^{4}+36{t}^{2}-72)}{({t}^{6}-6{t}^{4}+36{t}^{2}+72)^{2}}}\\
\b_4(t)&=2-{\frac{8{t}^{2}({t}^{4}-20{t}^{2}+60)({t}^{8}+72{t}^{4}-2160)}{({t}^{8}-16{t}^{6}+120{t}^{4}+720)^{2}}},\\
\b_5(t)&=\tfrac52-{\frac{10({t}^{6}-30{t}^{4}+180{t}^{2}-120)({t}^{12}-12{t}^{10}+180{t}^{8}-480{t}^{6}-3600{t}^{4}-
43200{t}^{2}+43200)}{({t}^{10}-30{t}^{8}+360{t}^{6}-1200{t}^{4}+3600{t}^{2}+7200)^{2}}},
\end{split}\]
and the first few monic orthogonal polynomials are given by
\begin{align*}
P_1(x;t)&=x-{\frac {t ( {t}^{2}+6 ) }{2({t}^{2}+2)}},\\
P_2(x;t)&={x}^{2}-{\frac {t ( {t}^{4}+4{t}^{2}+12 )}{{t}^{4}+12}}x+{\frac {{t}^{6}+6{t}^{4}+36{t}^{2}-72}{4({t}^{4}+12)}},\\
P_3(x;t)&={x}^{3}-{\frac {3t ( {t}^{6}-2{t}^{4}+20{t}^{2}+120)}{2({t}^{6}-6{t}^{4}+36{t}^{2}+72)}}x^2+{\frac {
 3( {t}^{8}+40{t}^{4}-240 )}{4({t}^{6}-6{t}^{4}+36{t}^{2}+72)}}x 
 -{\frac { t( {t}^{8}+72{t}^{4}-2160 )}{8({t}^{6}-6{t}^{4}+36{t}^{2}+72)}},\\
P_4(x;t)&= {x}^{4}-{\frac {2t ({t}^{8}-12{t}^{6}+72{t}^{4}+240{t}^{2} +720)}{{t}^{8}-16{t}^{6}+120{t}^{4}+720}} {x}^{3} 
+{\frac { 3( {t}^{10}-10{t}^{8}+80{t}^{6}+1200{t}^{2}-2400)}{2({t}^{8}-16{t}^{6}+120{t}^{4}+720)}} {x}^{2}\\ 
&\phantom{=x^4\ }-{\frac {t( {t}^{10}-10{t}^{8}+120{t}^{6}-240{t}^{4}-1200{t}^{2}-7200 )}{2({t}^{8}-16{t}^{6}+120{t}^{4}+720)}}x\\ 
&\phantom{=x^4\ }+{\frac {{t}^{12}-12{t}^{10}+180{t}^{8}-480{t}^{6}-3600{t}^{4}-43200{t}^{2}+43200}{16({t}^{8}-16{t}^{6}+120{t}^{4}+720)}},\\
P_5(x;t)&={x}^{5}-{\frac { 5t(t^{10}-26t^{8}+264t^{6}-336t^{4}+1680t^{2}+10080 )}{2(t^{10}-30t^{8}+360
t^{6}-1200t^{4}+3600t^{2}+7200)}}x^4\\
&\phantom{={x}^{5}\ }+{\frac {5 ( t^{12}-24t^{10}+252t^{8}-672t^{6}+5040t^{4}-20160)}{2(t^{10}-30t^{8}+360t^{6}-1200t^{4}+3600t^{2}+7200)}}x^3\\
&\phantom{={x}^{5}\ }-{\frac { 5t( t^{12}-24t^{10}+300t^{8}-1440t^{6}+5040t^{4}-100800 )}
{4(t^{10}-30t^{8}+360t^{6}-1200t^{4}+3600t^{2}+7200)}}x^2\\
&\phantom{={x}^{5}\ }+{\frac { 5( t^{14}-26t^{12}+396t^{10}-2520t^{8}+5040t^{6}-50400t^{4}-100800t^{2}+201600 )}
{16(t^{10}-30t^{8}+360t^{6}-1200t^{4}+3600t^{2}+7200)}}x\\
&\phantom{={x}^{5}\ }-{\frac { t( t^{14}-30t^{12}+540t^{10}-4200t^{8}+10800t^{6}-151200t^{4}-504000t^{2}+3024000)}
{32(t^{10}-30t^{8}+360t^{6}-1200t^{4}+3600t^{2}+7200)}}.
\end{align*}

\def\ams{American Mathematical Society}
\def\AAM{Acta Appl. Math.}
\def\ARMA{Arch. Rat. Mech. Anal.}
\def\bull{Acad. Roy. Belg. Bull. Cl. Sc. (5)}
\def\AC{Acta Crystrallogr.}
\def\AM{Acta Metall.}
\def\ampa{Ann. Mat. Pura Appl. (IV)}
\def\AP{Ann. Phys., Lpz.}
\def\APNY{Ann. Phys., NY}
\def\APP{Ann. Phys., Paris}
\def\BAMS{Bull. Amer. Math. Soc.}
\def\CJP{Can. J. Phys.}
\def\cmp{Commun. Math. Phys.}
\def\CMP{Commun. Math. Phys.}
\def\cpam{Commun. Pure Appl. Math.}
\def\CPAM{Commun. Pure Appl. Math.}
\def\CQG{Classical Quantum Grav.}
\def\crp{C.R. Acad. Sc. Paris}
\def\CSF{Chaos, Solitons \&\ Fractals}
\def\DE{Diff. Eqns.}
\def\DU{Diff. Urav.}
\def\ejam{Europ. J. Appl. Math.}
\def\EJAM{Europ. J. Appl. Math.}
\def\funk{Funkcial. Ekvac.}
\def\FUNK{Funkcial. Ekvac.}
\def\IP{Inverse Problems}
\def\JAMS{J. Amer. Math. Soc.}
\def\JAP{J. Appl. Phys.}
\def\JCP{J. Chem. Phys.}
\def\JDE{J. Diff. Eqns.}
\def\JFM{J. Fluid Mech.}
\def\JJAP{Japan J. Appl. Phys.}
\def\JP{J. Physique}
\def\JPhCh{J. Phys. Chem.}
\def\JMAA{J. Math. Anal. Appl.}
\def\JMMM{J. Magn. Magn. Mater.}
\def\JMP{J. Math. Phys.}
\def\jmp{J. Math. Phys}
\def\JNMP{J. Nonl. Math. Phys.}
\def\jpa{J. Phys. A}
\def\JPA{J. Phys. A}
\def\JPB{J. Phys. B: At. Mol. Phys.} 
\def\jpb{J. Phys. B: At. Mol. Opt. Phys.} 
\def\JPC{J. Phys. C: Solid State Phys.} 
\def\JPCM{J. Phys: Condensed Matter} 
\def\JPD{J. Phys. D: Appl. Phys.}
\def\JPE{J. Phys. E: Sci. Instrum.}
\def\JPF{J. Phys. F: Metal Phys.}
\def\JPG{J. Phys. G: Nucl. Phys.} 
\def\jpg{J. Phys. G: Nucl. Part. Phys.} 
\def\JSP{J. Stat. Phys.}
\def\JOSA{J. Opt. Soc. Am.}
\def\JPSJ{J. Phys. Soc. Japan}
\def\JQSRT{J. Quant. Spectrosc. Radiat. Transfer}
\def\LMP{Lett. Math. Phys.}
\def\LNC{Lett. Nuovo Cim.}
\def\NC{Nuovo Cim.}
\def\NIM{Nucl. Instrum. Methods}
\def\NL{Nonlinearity}
\def\NMJ{Nagoya Math. J.}
\def\NP{Nucl. Phys.}
\def\pl{Phys. Lett.}
\def\PL{Phys. Lett.}
\def\PMB{Phys. Med. Biol.}
\def\PR{Phys. Rev.}
\def\PRL{Phys. Rev. Lett.}
\def\PRS{Proc. R. Soc.}
\def\prsl{Proc. R. Soc. Lond. A}
\def\PRSL{Proc. R. Soc. Lond. A}
\def\PS{Phys. Scr.}
\def\PSS{Phys. Status Solidi}
\def\PTRS{Phil. Trans. R. Soc.}
\def\RMP{Rev. Mod. Phys.}
\def\RPP{Rep. Prog. Phys.}
\def\RSI{Rev. Sci. Instrum.}
\def\SAM{Stud. Appl. Math.}
\def\sam{Stud. Appl. Math.}
\def\SSC{Solid State Commun.}
\def\SST{Semicond. Sci. Technol.}
\def\SUST{Supercond. Sci. Technol.}
\def\ZP{Z. Phys.}
\def\JCAM{J. Comput. Appl. Math.}

\def\OUP{Oxford University Press}
\def\CUP{Cambridge University Press}
\def\AMS{American Mathematical Society}
\def\refpp#1#2#3#4#5{\vspace{-0.25cm}
\bibitem{#1} \textrm{\frenchspacing#2}, \textit{#4}, #3 (#5).}

\def\refjl#1#2#3#4#5#6#7{\vspace{-0.25cm}
\bibitem{#1} \textrm{\frenchspacing#2}, \textit{#6}, 
{\frenchspacing#3},\ \textbf{#4} (#7) #5.}

\def\refjltoap#1#2#3#4#5#6#7{\vspace{-0.25cm}
\bibitem{#1} \textrm{\frenchspacing#2}, \textit{#6}, 
\textrm{\frenchspacing#3} (#7), #5.} 

\def\refbk#1#2#3#4#5{\vspace{-0.25cm}
\bibitem{#1} \textrm{\frenchspacing#2}, \textit{#3}, #4, #5.} 

\def\refbkk#1#2#3#4#5{\vspace{-0.25cm}
\bibitem{#1} \textrm{\frenchspacing#2}, \textit{#3}\ #4, #5.} 

\def\refcf#1#2#3#4#5#6{\vspace{-0.25cm}
\bibitem{#1} \textrm{\frenchspacing#2}, \textit{#3},
in: \textit{#4}, {\frenchspacing#5}, #6.}

\end{document}